\documentclass[a4paper,UKenglish,cleveref, autoref, thm-restate]{lipics-v2021}

\usepackage{mathtools}
\usepackage{amsfonts}
\usepackage{marvosym}
\usepackage{xspace}
\usepackage{nicefrac}

\usepackage{thmtools}
\usepackage{thm-restate}

\usepackage{algorithm}
\usepackage{algpseudocode}

\usepackage{tikz}
\usetikzlibrary{arrows,decorations.pathmorphing,positioning,fit,trees,shapes,shadows,automata,calc}
\usetikzlibrary{arrows.meta}
\tikzset{
  plainNodes/.style={draw=black, thick, circle, font=\scriptsize, minimum size=0.78cm},
  hatNodes/.style={draw=blue!70!black, thick, circle, font=\scriptsize, fill=blue!10, minimum size=0.78cm},
  winNode/.style={draw=green!70!black, thick, circle, font=\tiny, fill=green!10, minimum size=0.78cm},
  loseNode/.style={draw=red!70!black, thick, circle, font=\tiny, fill=red!10, minimum size=0.78cm},
  labelNodes/.style={draw=none, rectangle, font=\scriptsize, inner sep=3pt},
  tinylabelNodes/.style={draw=none, rectangle, font=\tiny, inner sep=3pt},
  deletedTrans/.style = {draw=orange!70,thick},
  addedTrans/.style = {draw=purple!70,thick},
  EveNodes/.style = {draw=black, regular polygon, regular polygon sides=4, thick, font=\scriptsize, minimum size=0.78cm},
  AdamNodes/.style = {draw=black, regular polygon, regular polygon sides=5, thick, font=\scriptsize, minimum size=0.78cm},
}
\tikzset{>=stealth, every path/.style={semithick}}
\usepackage{pgfplots}
\pgfplotsset{compat=newest}
\usepackage{float} 
\usepackage{subcaption}
\usepackage{rotating}
\usepackage{multicol}
\usepackage{tablefootnote}

\usepackage{url}
\usepackage{appendix}
\usepackage{footnote}
\usepackage{tablefootnote}
\usepackage{twoopt}
\usepackage{pdfpages}

\usepackage{pgffor}
\usepackage{comment}

\usepackage{eucal}

\usepackage{letltxmacro}

\newlength{\rulethickness}
\colorlet{goodArea}{green!50}
\colorlet{badArea}{red!50}

\tikzset{
state/.style={draw,circle, inner sep=0pt, minimum size=20pt},
action/.style={draw, circle,  inner sep=0.75pt, fill},
probabilistic/.style={->},
markovian/.style={dashed, ->},
}

\newcommand{\mc}{\ensuremath{\mathcal{M}}}
\newcommand{\omc}{\ensuremath{\widetilde{\mathcal{M}}}}
\newcommand{\mcprob}{\ensuremath{\mathbf{\delta}}}
\newcommand{\omcprob}{\ensuremath{\mathbf{\widetilde{\delta}}}}
\renewcommand{\Pr}{\ensuremath{\mathrm{Pr}}}

\newcommand{\evestrategy}{\ensuremath{\Sigma^A_{\exists}}}
\newcommand{\adamstrategy}{\ensuremath{\Sigma^A_{\forall}}}
\newcommand{\Veve}{\ensuremath{V_{\exists}}}
\newcommand{\Vadam}{\ensuremath{V_{\forall}}}

\newcommand{\arntrans}{\ensuremath{\Delta}}

\newcommand{\arnpr}{\ensuremath{\mathbb{P}}}
\newcommand{\oarnpr}{\ensuremath{\widetilde{\mathbb{P}}}}

\newcommand{\greach}{\ensuremath{\mathit{RE}}}
\newcommand{\gparity}{\ensuremath{\mathit{PA}}}

\newcommand{\vE}{\ensuremath{\langle{E}\rangle}}
\newcommand{\vA}{\ensuremath{\langle{A}\rangle}}
\newcommand{\vwin}{\ensuremath{v_{\textit{win}}}}
\newcommand{\vlose}{\ensuremath{v_{\textit{lose}}}}

\newcommand{\oVeve}{\ensuremath{\overline{V}_{\exists}}}
\newcommand{\oVadam}{\ensuremath{\overline{V}_{\forall}}}

\newcommand{\ov}{\ensuremath{{\overline{v}}}}
\newcommand{\ou}{\ensuremath{{\overline{u}}}}

\newcommand{\oV}{\ensuremath{{\overline{V}}}}
\newcommand{\oU}{\ensuremath{{\overline{U}}}}
\newcommand{\oC}{\ensuremath{{\overline{C}}}}

\newcommand{\wG}{\ensuremath{\widetilde{G}}}
\newcommand{\wV}{\ensuremath{\widetilde{V}}}
\newcommand{\wu}{\ensuremath{\widetilde{u}}}
\newcommand{\wv}{\ensuremath{\widetilde{v}}}
\newcommand{\ww}{\ensuremath{\widetilde{w}}}
\newcommand{\warntrans}{\ensuremath{{\widetilde{\arntrans}}}}
\newcommand{\wE}{\ensuremath{\widetilde{E}}}
\newcommand{\wPr}{\ensuremath{\widetilde{\mathrm{Pr}}}}
\newcommand{\wU}{\ensuremath{\widetilde{U}}}

\newcommand{\Reach}{\ensuremath{\mathit{Reach}}}
\newcommand{\crossPath}{\ensuremath{\mathit{crossPath}}}
\newcommand{\enterEven}{\ensuremath{\mathit{enterEven}}}
\newcommand{\enterOdd}{\ensuremath{\mathit{enterOdd}}}
\newcommand{\winEven}{\ensuremath{\mathit{winEven}}}
\newcommand{\winOdd}{\ensuremath{\mathit{winOdd}}}

\newcommand{\pM}{\ensuremath{\delta_{\textit{min}}}}

\newcommand{\bigO}{\ensuremath{\mathcal{O}}}

\newcommand{\aK}{\ensuremath{\alpha_{k}}}
\newcommand{\aKP}{\ensuremath{\alpha_{k{+}1}}}

\newcommand{\epsi}{\ensuremath{(n!)^2 M^{2n^2}}}

\title{A Direct Reduction from Stochastic Parity Games to Simple Stochastic Games} 

\author{Rapha\"{e}l Berthon}{RWTH Aachen University, Germany}{berthon@cs.rwth-aachen.de}{https://orcid.org/0000-0002-2580-5193}{Funded by DFG Project POMPOM (KA 1462/6-1)}

\author{Joost-Pieter Katoen}{RWTH Aachen University, Germany}{katoen@cs.rwth-aachen.de}{https://orcid.org/0000-0002-6143-1926}{}

\author{Zihan Zhou}{National University of Singapore, Singapore}{zihan.zhou@u.nus.edu}{https://orcid.org/0009-0009-7406-6663}{}

\authorrunning{R. Berthon, J-P. Katoen and Z. Zhou} 
\Copyright{Rapha\"el Berthon, Joost-Pieter Katoen, Zihan Zhou}

\ccsdesc[500]{Theory of computation~Probabilistic computation} 

\keywords{stochastic games, parity, reduction} 

\category{} 

\relatedversion{} 

\nolinenumbers 

\begin{document}

\maketitle

\begin{abstract}
Significant progress has been recently achieved in developing efficient solutions for simple stochastic games (SSGs), focusing on reachability objectives. While reductions from stochastic parity games (SPGs) to SSGs have been presented in the literature through the use of multiple intermediate game models, a direct and simple reduction has been notably absent. This paper introduces a novel and direct polynomial-time reduction from quantitative SPGs to quantitative SSGs. By leveraging a gadget-based transformation that effectively removes the priority function, we construct an SSG that simulates the behavior of a given SPG. We formally establish the correctness of our direct reduction. Furthermore, we demonstrate that under binary encoding this reduction is polynomial, thereby directly corroborating the known $\textbf{NP}\,\mathbf{\cap}\,\textbf{coNP}$ complexity of SPGs and providing new understanding in the relationship between parity and reachability objectives in turn-based stochastic games.

\end{abstract}

\section{Introduction}\label{ch:Introduction}

\emph{Stochastic games} (SGs) are a broadly used framework for decision-making under uncertainty with a well-defined objective. They aim at modeling two important aspects under which sequential decisions must be made: turn-based interaction with an adversary environment, and dealing with randomness. Stochastic games have been introduced long ago by Shapley~\cite{shapley1953stochastic} as important models, and many variations have been shown to be inter-reducible and in $\textbf{NP}\,\mathbf{\cap}\,\textbf{coNP}$~\cite{condon1992complexity,andersson2009complexity,Chatterjee_2011}. 
SGs find their applications in various fields, including artificial intelligence~\cite{leslie2020best}, economics~\cite{amir2003stochastic}, operations research~\cite{darlington2023stochastic}, or providing tools to graph theory~\cite{DBLP:journals/tcs/SimardDL21}. 
Moreover, \textit{Markov Decision Processes} (MDP)~\cite{puterman2014markov}, a foundational framework for modeling decision-making in stochastic environments, are special cases of SGs, where one of the players has no states under control. 

While we make use of simple stochastic games (SSGs) with reachability objectives, we focus on the specific case of stochastic \emph{parity} games (SPGs), which are zero-sum, and where the set of winning runs is $\omega$-regular. Solving such a game consists in finding an optimal strategy and determining its winning probability.

Take, for instance, the case of a market competition, where two firms Alpha and Beta try to expand their market share. From the standpoint of each firm, the other acts as a direct competitor, and therefore we assume the players are \textit{adversarial}. States represent the relative valuation of the firms over a sustained period of time, and business decisions (made by either firm Alpha or Beta) and the fluctuation of market shares, policy changes, and external forces (modeled as randomness) lead to transitions between game states. Firm Alpha is interested in keeping its share above a set key threshold, say 50\%. We distinguish between three priorities, leading to a \textit{zero-sum} game: 
\begin{itemize}
    \item \emph{Priority $0$}: Alpha's market share is significantly above 50\%.
    \item \emph{Priority $1$}: Alpha's market share is significantly below 50\%. 
    \item \emph{Priority $2$}: Alpha's market share fluctuates around 50\%. While this is sustainable without major fluctuation, this is not sustainable if the only other fluctuation is Alpha's share regularly dropping below 50\% (priority $1$).  
\end{itemize}
If the minimum priority visited infinitely often is $0$, Alpha can manage in the long term to regularly dominate the market, recovering any loss that occurred in the meantime. If the minimum priority visited infinitely often is $1$, despite any temporary success, Alpha's market share will stay near or below 50\% in the long run, which is not sustainable for the firm. Finally, if the minimum priority visited infinitely often is $2$, Alpha and Beta will eventually find an even balance point.

A variety of algorithms have been considered for (reachability) SSGs~\cite{condon1990algorithms,eisentraut2022value,azeem2022optimistic,phalakarn2020widest}, which we present in our related work section below. 
In Markov chains (MCs), $\omega$-regular objectives can be reduced directly to reachability objectives~\cite{baier2008principles}. A similar reduction exists from MDPs with $\omega$-regular objectives to reachability objectives and has been used extensively, for example in~\cite{DBLP:journals/lmcs/EtessamiKVY08}. This means that solvers often focus on optimizing specifically the computation of reachability probabilities. Such a direct reduction is lacking for SPGs. To reduce quantitative SPGs to SSGs, some intermediate steps are necessary, via a reduction to stochastic mean-payoff and stochastic discounted-payoff games~\cite{chatterjee2008reduction,andersson2009complexity} (see the lower part of Figure~\ref{fig:intro1}), making this approach less appealing. For qualitative solutions, a translation via deterministic parity games (i.e. with no random states) exists~\cite{chatterjee2003sspg, chatterjee2004quantitative,Chatterjee_2011}, see the upper part of Figure~\ref{fig:intro1}.

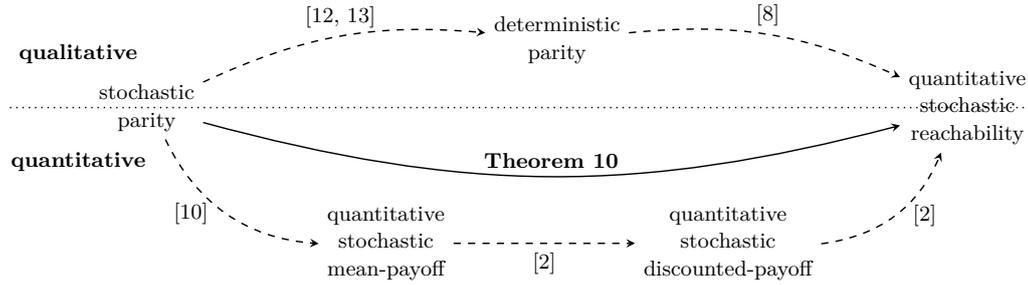
\begin{figure}
\centering\scalebox{0.90}{
\begin{tikzpicture}
[myph/.style={
    rectangle,
    minimum width=0.1cm,
},
every node/.append style={
    font=\small,
},
]
\node[myph,align=center] (00) at (-3,0.8) {\textbf{qualitative}};
\node[myph,align=center] (01) at (-3,-0.8) {\textbf{quantitative}};
\node[myph,align=center] (0) at (-2,0) {stochastic \\ parity};
\node[myph,align=center] (4) at (10,0) {quantitative \\ stochastic \\ reachability};

\node[myph,align=center] (1) at (4,1) {deterministic \\ parity};

\node[myph,align=center] (2) at (1.5,-2) {quantitative \\ stochastic \\ mean-payoff};
\node[myph,align=center] (3) at (6.5,-2) {quantitative \\ stochastic \\ discounted-payoff};

\draw[dotted] (-4,0) -- (11,0);

\draw[->,dashed,] (0) to[bend left=15] node[above] {\cite{chatterjee2003sspg, chatterjee2004quantitative}} (1);
\draw[->,dashed] (1) to[bend left=15] node[above] {\cite{Chatterjee_2011}} (4);
\draw[->,dashed] (0) to[bend right] node[left,xshift=-4pt] {\cite{chatterjee2008reduction}} (2);
\draw[->,dashed] (2) -- node[below] {\cite{andersson2009complexity}} (3);
\draw[->,dashed] (3) to[bend right] node[right,xshift=4pt] {\cite{andersson2009complexity}} (4);
\draw[->] (0) to[bend right=15] node[above] {\textbf{Theorem~\ref{th:reduction}}} (4);
\end{tikzpicture}
}
\caption{Reducing SPGs to SSGs}
\label{fig:intro1}
\end{figure}






\paragraph*{Outline and Contribution}
In this paper, we propose a \emph{direct} reduction from SPGs to SSGs with a reachability objective (see the solid arrow in Figure~\ref{fig:intro1}). To that end, we leverage a gadget whose structure comes from~\cite{Chatterjee_2011}, but where we use new probability values, to reduce deterministic parity games to quantitative SSGs. 
Given an SPG $G$, where a parity condition is satisfied if the minimum priority seen infinitely often is even, we show in Section~\ref{subsec:gadgets} how to use the gadget to transform $G$ into an SSG $\wG$. 
This introduces two new sink states, one winning and the other losing for the reachability objective. Parity values are removed, and every transition going to a state that used to be even (respectively odd) now has a small chance to go to the winning (resp. losing) sink. We scale the probabilities, with lower parity values yielding a higher chance to go to a sink. 
Theorem~\ref{th:reduction} ensures that any optimal strategy in $\wG$ is also optimal in $G$ (the reciprocal may not be true). We can then compute an optimal memoryless strategy in $\wG$, and compute its value in $G$.  

We show in Theorem~\ref{th:size} that under binary encoding, our reduction is polynomial. We thus reobtain in a direct way the classical result that solving quantitative SPGs is in the complexity class $\textbf{NP}\,\mathbf{\cap}\,\textbf{coNP}$~\cite{andersson2009complexity,Chatterjee_2011}. While the complexity remains the same as for existing algorithms, and the values used in the reduction make it unlikely to be very efficient in practice, this new approach implies that any efficient SSG solver can be used for SPGs. The direct reduction was already conjectured to exist by Chatterjee and Fijalkow in~\cite{Chatterjee_2011}, but as expected, proving its correctness is challenging, and involves the computations of very precise probability bounds. Despite the inspiration drawn from a known gadget, the technical depth of this paper resides in the intricate and novel proofs for the correctness of our reduction. In addition, our direct reduction gives new insights into the relationship between SSGs and SPGs.  

In Section~\ref{ch:preliminaries}, we present the necessary background knowledge. In Section~\ref{sec:gadegts}, we define the gadget and present some related results, which we use in Section~\ref{sec:reduction} to define the reduction properly, and to show its correctness and its complexity. We give some concluding remarks in Section~\ref{ch:conclusion}.

\paragraph*{Related Work}

Stochastic parity games, mean-payoff games and discounted payoff games can all be reduced to SSGs~\cite{jurdzinski1998deciding,zwick1996complexity}, and this also applies to their stochastic extensions, namely stochastic parity games~\cite{Chatterjee_2011}, stochastic mean payoff games and stochastic discounted payoff games~\cite{andersson2009complexity}. SSGs also find their applications in the analysis of MDPs, serving as abstractions for large MDPs~\cite{kattenbelt2010game}. The amount of memory required to solve stochastic parity games has been studied in~\cite{DBLP:journals/lmcs/BouyerORV23}.

Various extensions have been considered within this family of inter-reducible stochastic games. Introducing more than two players allows for the analysis of Nash equilibria~\cite{DBLP:conf/csl/ChatterjeeMJ04,DBLP:journals/corr/abs-1109-4017}. Using continuous states can provide tools to represent timed systems~\cite{DBLP:conf/adhs/MajumdarMSS21}. Multi-objective approaches have been employed to synthesize systems that balance average expected outcomes with worst-case guarantees~\cite{DBLP:conf/concur/ChatterjeeP19}. Parity objectives are significant in many of these scenarios where long-run behavior is relevant, but the classical reduction to SSGs cannot be directly applied. 

Common approaches to solving SSGs, as presented in~\cite{condon1990algorithms}, include value iteration (VI), strategy iteration (SI), and quadratic programming, but are all exponential in the size of the SSG. These approaches have been widely studied on MDPs, where recent advancements have been made to apply VI with guarantees, using interval VI~\cite{baier2017ensuring}, sound VI~\cite{quatmann2018sound}, and optimistic VI~\cite{hartmanns2020optimistic}. Interestingly, optimistic VI does not require an a priori computation of starting vectors to approximate from above. 
Similar ideas have been lifted to SSGs: Eisentraut et al.~\cite{eisentraut2022value} introduce a VI algorithm for under- and over-approximation sequences, as well as the first practical stopping criterion for VI on SSGs. Optimistic VI has been adapted to SSGs~\cite{azeem2022optimistic}, and a novel bounded VI with the concept of widest path has been introduced in~\cite{phalakarn2020widest}. A comparative analysis~\cite{kvretinsky2022comparison} suggests VI and SI are more efficient. Storm~\cite{hensel2021probabilistic} and PRISM~\cite{kwiatkowska2011prism} are two popular model checkers incorporating different variants of VI and SI, and both employ VI as the default algorithm for solving MDPs. PRISM-games~\cite{kwiatkowska2020prism} exploits VI for solving SSGs.

For SPGs, we distinguish three main approaches. Chatterjee et al.~\cite{chatterjee2006strategy} use a strategy improvement algorithm requiring randomized sub-exponential time. With $n$ game states and $d$ priorities, the expected running time is in $2^{O(\sqrt{dn\log(n)})}$. 
The probabilistic game solver GIST~\cite{chatterjee2010gist} reduces qualitative SPGs to deterministic parity games (DPG), and benefits from several quasi-polynomial algorithms for DPGs~\cite{jurdzinski2017succinct,parys2019parity,lehtinen2018modal} since the breakthrough made by Calude et al.~\cite{calude2017deciding}, but this approach is unlikely to achieve polynomial running time~\cite{czerwinski2019universal}. Hahn et al.~\cite{hahn2020model} reduce SPGs to SSGs, allowing the use of reinforcement learning to approximate the values without knowing the game's probabilistic transition structure. Their reduction is only proven correct in the limit.

\section{Preliminaries}\label{ch:preliminaries}
Our notations on Markov chains and stochastic games on graphs mainly come from~\cite{baier2008principles}.

\subsection{Discrete-Time Markov Chains}
A \textit{discrete distribution} over a countable set $\mathcal{A}$ is a function $\mu:\mathcal{A}\xrightarrow{}\mathbb{R}_{\geq 0}$ with $\sum_{a\in \mathcal{A}}\mu(a) = 1$. 
The \textit{support} of the discrete distribution $\mu$ is defined as $supp(\mu)\triangleq\{a \in \mathcal{A}\ |\ \mu(a)>0\}$. 
We denote the set of all discrete distributions over $\mathcal{A}$ with $\mathbb{D}(\mathcal{A})$.

A \textit{discrete-time Markov Chain (MC)} $\mc$ is a tuple $\mc = (V, \mcprob, v_{I})$ where $V$ is a finite set of states, $\mcprob: V \rightarrow \mathbb{D}(V)$ is a probabilistic transition function, and $v_{I}\in V$ is the initial state. 
Given $\delta(v) = \mu$ with $\mu(v') = p$, we write $\delta(v, v') = p$. 
For $S\subseteq V$ and $v\in V$, let $\delta(v,S)=\sum_{s\in S} \delta(v,S)$.

An infinite sequence $\pi = v_0 v_1 \dots \in V^{\omega}$ is an \textit{infinite path} through MC $\mc$ if $\mcprob(v_i,v_{i+1})>0$ for all $i\in \mathbb{N}$. We denote all infinite paths that start from state $v\in V$ with $\textit{Paths}(v)$. 
\textit{Prefixes} of infinite path $\pi = v_0 v_1 \dots \in V^{\omega}$ are $\{v_0 \cdots v_i\ |\ i\in\mathbb{N}\}$ and are \textit{finite} paths. We denote all finite paths that start from state $v\in V$ with $\textit{Paths}^*(v)$.
The set of infinitely often visited states in $\pi = v_0 v_1 \dots \in V^{\omega}$ is defined as $\textit{inf}(\pi) = \{v \in V\ |\ \forall n \in \mathbb{N}, \exists k \in \mathbb{N} \textit{ s.t. } v_{n+k} = v \}$.

The probability $\Pr$ of a finite path $\pi = v_0 v_1 \dots v_n \in V^*$ is given by $\prod_{i\in[0,{n-1}]}\mcprob(v_{i},v_{i+1})$. The set of infinite paths that start with a given finite path is called a \textit{cylinder}, and as in~\cite{baier2008principles}, we extend the probability of cylinders in a unique way to all \emph{measurable sets} of $V^{\omega}$.

\paragraph*{Reachability Probabilities}

Let $\mc=(V, \mcprob, v_{I})$ be an MC. For target states $T \subseteq V$ and starting state $v_0 \in V$, the event of reaching $T$ is defined as
$ \Reach(T) =  \{ v_0 v_1 \dots \in V^{\omega}\ |\ \exists i \in \mathbb{N}, v_i \in T \}$.
The probability to reach $T$ from $v_0$ is defined as 
$ \Pr^{v_0}(\Reach(T)) = \Pr(\{\hat{\pi}\ |\ \hat{\pi} \in \textit{Paths}^*(v_0) \cap ((V \backslash T)^* T)\})$.

Let variable $x_{v}$ denote the probability of reaching $T$ from any $v \in V$. 
Whether $T$ is reachable from a given state $v$ can be determined using standard graph analysis.
Let  $\textit{Pre}^*(T)$ denote the set of states from which $T$ is reachable. If $v \notin \textit{Pre}^*(T)$, then $x_{v}=0$. If $v \in T$, then $x_{v}=1$. Otherwise, 
$
x_{v}=\sum_{u \in V\backslash T} \mcprob(v,u) \cdot x_{u} + \sum_{w\in T}\mcprob(v,w).
$
This is equivalent to a linear equation system, formalized as follows:
\begin{theorem}[Reachability Probability of Markov Chains~\cite{baier2008principles}]\label{th:reachability}
    Given MC $\mc=(V , \mcprob, v_{I})$ and  target states $T \subseteq V$,
    let $V_{Q} = \textit{Pre}^*(T) \backslash T$, $\mathbf{A} = (\mcprob(v, v'))_{v,v' \in V_{Q}}$ and $\mathbf{b} = (b_{v})_{v\in V_{Q}} = (\delta(v,T))_{v\in V_{Q}}$.
    Then, the vector $\mathbf{x} = (x_{v})_{v\in V_{Q}}$ with $x_{v}=\Pr^{v}(\Reach(T))$ is the \textbf{unique} solution of the linear equation system
    $\mathbf{x} = \mathbf{A} \cdot \mathbf{x} + \mathbf{b}$.
\end{theorem}

\paragraph*{Limit Behavior}

Let MC $\mc=(V, \mcprob, v_{I})$. A set $L \subseteq V$ is \textit{strongly connected} if for all pairs of states $v,v' \in L$, $v$ and $v'$ are mutually reachable. Hence a singleton $\{v\}$ is strongly connected if $\mcprob(v,v)>0$.
Set $L \subseteq V$ is a \textit{strongly connected component (SCC)} if it is maximally strongly connected, i.e., there does not exist another set $L' \subseteq V$ and $L\subsetneq L'$ such that $L'$ is strongly connected.
$L \subseteq V$ is a \textit{bottom SCC (BSCC)} if $L$ is a SCC and there is no transition leaving $L$, i.e., there does not exist $v \in L, v' \in V\backslash L$ such that $\mcprob(v, v')>0$.
We denote the set of BSCCs in MC $\mc$ with $\textit{BSCC}(\mc)$.

The limit behavior of an MC regarding the infinitely often visited states is captured by the following theorem.
\begin{theorem}[Limit behavior of Markov Chains~\cite{baier2008principles}]\label{th:limitMarkov}
For MC $\mc=(V, \mcprob, v_{I})$, it holds that $\Pr\{\pi \in \textit{Paths}(v_I)\  |\ \textit{inf}(\pi) \in \textit{BSCC}(\mc)\} = 1$.
\end{theorem}

\subsection{Stochastic Games}

A \textit{stochastic arena} $G$ is a tuple $G=((V,E), (\Veve, \Vadam, V_{R}), \arntrans)$, where $(V,E)$ is a directed graph,  with a finite set of vertices $V$, partitioned as $\Veve \uplus \Vadam \uplus V_R = V$, and a set of edges $E \subseteq V\times V$. The probabilistic transition function $\arntrans$ is such that for all $v_r\in V_R$, $\arntrans(v_r)$ is a distribution over $V$, and for $v\in \Veve \uplus \Vadam$, $v'$, we have $(v_r,v) \in E$ if and only if $v\in supp(\arntrans(v_r))$. We usually uncurry $\arntrans(v_r)(v)$ and write $\arntrans(v_r,v)$ .

Without loss of generality, we assume each vertex has at least one successor. This property is called \textit{non-blocking}.
The finite set $V$  of vertices is partitioned into three sets: $\Veve$ --- vertices where Eve chooses the successor, $\Vadam$ --- vertices where Adam chooses the successor, and $V_R$ are the random vertices. A stochastic arena is a Markov Decision Process (MDP) if either $\Veve = \emptyset$ or $\Vadam = \emptyset$, and an MC if both $\Veve = \emptyset$ and $\Vadam = \emptyset$.

Figure~\ref{fig:running-example-G} illustrates a stochastic arena $G$. Square-shaped vertex \(v_3\) is a vertex in \(\Veve\) where Eve chooses the successor, pentagon-shaped vertex \(v_4\) is in \(\Vadam\) where Adam chooses the successor, and the circular vertices \(V_R = \{v_0, v_1, v_2, v_5\}\) are random. Edges from random vertices are annotated with probabilities from \(\arntrans\).
\begin{figure}
\centering
\begin{tikzpicture}[scale=2]
\node[plainNodes] (v0) at (0,0) {$v_0$};
\node[plainNodes] (v1) at (1,0) {$v_1$};
\node[plainNodes] (v2) at (2,0) {$v_2$};
\node[AdamNodes] (v3) at (3,0) {$v_4$};
\node[EveNodes] (v4) at (3,0.75) {$v_3$};
\node[plainNodes] (v5) at (4,0.75) {$v_5$};

\draw[->] (0,-0.5) -- (v0);
\draw[->] (v0) -- node[labelNodes, below]{$0.1$} (v1);
\draw[->] (v1) -- node[labelNodes, below]{$0.1$} (v2);
\draw[->] (v2) -- node[labelNodes, below]{$0.1$} (v3);
\draw[->] (v3) -- (v4);
\draw[->] (v4) -- ++(0,0.5)  -| (v5);
\draw[->] (v5) -- node[labelNodes, above]{$0.5$} (v4);
\draw[->] (v5) -- node[labelNodes, left]{$0.5$} (v3.north east);
\draw[->] (v3) -| (v5);
\draw[->] (v3) -- ++(0,-0.5) -| (v1);
\draw[->] (v1) -- ++(0,0.3) -| node[labelNodes, above, pos=0.2]{$0.9$} (v0.north east);
\draw[->] (v2) -- ++(0,0.6) -| node[labelNodes, above, pos=0.2]{$0.9$} (v0.north);
\draw[->] (v4) -| (v2.north east);
\draw[->] (v0) edge[loop left] node[labelNodes, left]{$0.9$} (v0);

\end{tikzpicture}
\caption{An example stochastic arena $G$}
\label{fig:running-example-G}
\end{figure}
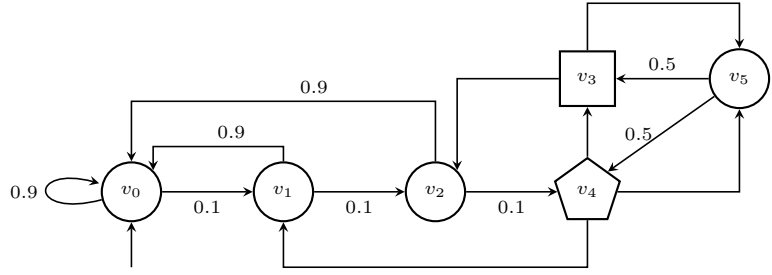

\paragraph*{Strategies}
Let $G=((V,E), (\Veve, \Vadam, V_{R}), \arntrans)$ be a stochastic arena.
A \textit{strategy} $\sigma$ of Eve is a function $\sigma: V^*\cdot \Veve \rightarrow \mathbb{D}(V)$, such that for all $v_0 v_1 \dots v_n \in  V^*\cdot \Veve$, we have 
        $\sigma(v_0 v_1 \dots v_n, v_{n+1}) > 0 $ implies $(v_n,v_{n+1})\in E$.
A strategy $\gamma$ of Adam is defined analogously.
We denote the sets of all strategies of Eve and Adam with $\evestrategy$ and $\adamstrategy$ respectively.

A strategy $\sigma$ of Eve is a \textit{pure memoryless strategy}, if for all  $w,w'\in V^*$ and $v\in \Veve$, $\sigma(w \cdot v)=\sigma(w' \cdot v)$ and the support of this distribution is a singleton. A pure memoryless strategy $\gamma$ of Adam is defined analogously. We denote the sets of pure memoryless strategies of Eve and Adam with $\Sigma_{\exists}$ and $\Sigma_{\forall}$ respectively.

In a stochastic arena $G$, when Eve and Adam follow pure memoryless strategies $\sigma\in\Sigma_{\exists}$ and $\gamma\in\Sigma_{\forall}$ respectively, the arena $G_{\sigma,\gamma} = ((V,E'), (\Veve,\Vadam,V_R),\arntrans)$ results. Here, the new edge set $E'$ is such that for all $u\in \Veve$, $(u,v) \in E'$ if and only if $\sigma(u)=v$, and for all $u\in \Vadam$, $(u,v) \in E'$ if and only if $\gamma(u)=v$. We refer to such arenas obtained by fixing pure memoryless strategies as sub-arenas.
In fact, given a fixed starting vertex $v_I \in V$, we often view the sub-arena $G_{\sigma,\gamma}$ as an MC 
$\mc_{\sigma,\gamma} = (V, \mcprob, v_{I})$, where the state space is the vertex set $V$ in $G$, and the transition function $\mcprob$ combines deterministic moves indicated by strategies $\sigma$ and $\tau$, and the transition function $\Delta$ defined on random vertices:
\[
\delta(u, v) = 
\begin{cases} 
\Delta(u, v) & \text{if } u \in V_{R}, \text{or } u \in \Veve, \sigma(u) = v \text{ or } u \in \Vadam, \gamma(u) = v \\
0 & \text{otherwise}
\end{cases}
\]

We continue with the stochastic arena $G$ from Figure~\ref{fig:running-example-G}. Fixing strategy \(\sigma = [v_3 \mapsto v_5]\) for Eve and \(\gamma = [v_4 \mapsto v_5]\) for Adam induces the sub-arena \(G_{\sigma, \gamma}\), as shown in Figure~\ref{fig:running-example-subG}.
\begin{figure}
\centering
\begin{tikzpicture}[scale=2]
\node[plainNodes] (v0) at (0,0) {$v_0$};
\node[plainNodes] (v1) at (1,0) {$v_1$};
\node[plainNodes] (v2) at (2,0) {$v_2$};
\node[AdamNodes] (v3) at (3,0) {$v_4$};
\node[EveNodes] (v4) at (3,0.75) {$v_3$};
\node[plainNodes] (v5) at (4,0.75) {$v_5$};

\draw[->] (0,-0.4) -- (v0);
\draw[->] (v0) -- node[labelNodes, below]{$0.1$} (v1);
\draw[->] (v1) -- node[labelNodes, below]{$0.1$} (v2);
\draw[->] (v2) -- node[labelNodes, below]{$0.1$} (v3);
\draw[->] (v4) -- ++(0,0.5)  -| node[labelNodes, below, pos=0.2]{$1$} (v5);
\draw[->] (v5) -- node[labelNodes, above]{$0.5$} (v4);
\draw[->] (v5) -- node[labelNodes, left]{$0.5$} (v3.north east);
\draw[->] (v3) -| node[labelNodes, above, pos=0.2]{$1$} (v5);
\draw[->] (v1) -- ++(0,0.3) -| node[labelNodes, above, pos=0.2]{$0.9$} (v0.north east);
\draw[->] (v2) -- ++(0,0.6) -| node[labelNodes, above, pos=0.2]{$0.9$} (v0.north);
\draw[->] (v0) edge[loop left] node[labelNodes, left]{$0.9$} (v0);

\end{tikzpicture}
\caption{Sub-arena $G_{\sigma,\gamma}$ induced by strategies $\sigma = [v_4\mapsto v_5]$ and $\gamma=[v_3 \mapsto v_5]$}
\label{fig:running-example-subG}
\end{figure}
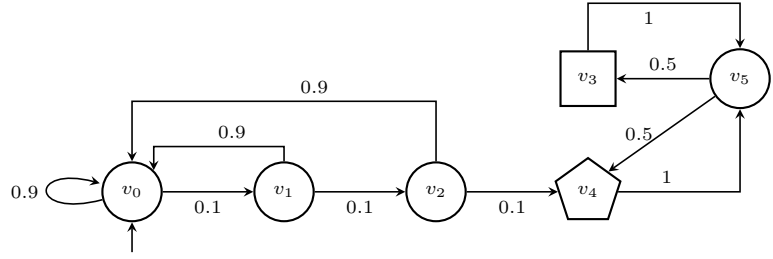

\paragraph*{Winning Objectives} 
A \textit{play} of $G$ is an infinite sequence of vertices $\pi=v_0 v_1\dots \in V^{\omega}$ where for all $i\in\mathbb{N}$, $(v_i,v_{i+1})\in E$. We denote the set of all plays of $G$ with $\Pi_G$, or in short $\Pi$ when $G$ is clear from the context. 

Let $G$ be a stochastic arena. A \textit{winning objective} for Eve is defined as a set of plays $\mathcal{O}\subseteq \Pi$. As we study zero-sum games, the winning objectives of the two players are complementary. The winning objective for Adam is thus $\Pi\backslash\mathcal{O}$. A play $\pi$ \textit{satisfies} an objective $\mathcal{O}$ if $\pi \in \mathcal{O}$, and is a \textit{winning play} of Eve. A winning play $\pi$ of Adam satisfies $\Pi\backslash\mathcal{O}$.

A \textit{reachability} objective asserts that the play in $G$ has to reach target vertices  $T \subseteq V$, formally given by $\greach(T)=\{v_0 v_1\dots \in\Pi\ |\ \exists k\in\mathbb{N}, v_k\in T\}$. 
If $T = \{v\}$ for some vertex $v$, we simply write $\greach(v)$. 

Let $p:V\rightarrow\mathbb{N}$ be a \textit{priority function} which assigns a priority $p(v)$ to each vertex $v\in V$. For $T\subseteq V$, let $p(T) = \{p(t)\ |\ t\in T\}$. A parity objective asserts that the minimum priority visited infinitely often along an infinite path is even:
$\gparity(p)=\{\pi=v_0 v_1\dots\in\Pi\ |\ \min(p(\textit{inf}(\pi)))\text{ is even}\}$.

We formally define stochastic games as follows:
\begin{definition}[Stochastic Games]\label{def:stochastic_games}
    Let $G=((V,E), (\Veve, \Vadam, V_{R}), \arntrans)$ be a stochastic arena. 
        A \textbf{stochastic game (SG)} with winning objective $\Phi\subseteq \Pi$ is defined as $(G,\Phi)$. If $\Phi$ is a reachability or parity objective, $(G,\Phi)$ is a \textit{stochastic reachability game (SRG)} or \textbf{stochastic parity game (SPG)} respectively. SRGs are also referred to as \textbf{simple stochastic games (SSG)}. 
    When the winning objective is clear from the context, we refer to $G$ as a stochastic game. 
\end{definition}

\paragraph*{Solving stochastic games}
Let $(G,\Phi)$ be an SG, and let Eve and Adam follow strategies $\sigma\in\evestrategy$ and $\gamma\in\adamstrategy$.    Given a starting vertex $v \in V$, the probability for play $\pi$ to satisfy $\Phi$ --- the probability for Eve to win --- is denoted $\arnpr_{\sigma,\gamma}^{v}(\Phi)$. The probability for Adam to win is $\arnpr_{\sigma,\gamma}^{v}(\Pi\backslash\Phi)$.

Let the \textit{value} of a vertex $v$ be the maximal probability of generating a play from $v$ that satisfies $\Phi$, formally defined using 
a \textit{value function} 
$\vE(\Phi)(v) = \sup_{\sigma\in\evestrategy}\inf_{\gamma\in\adamstrategy}\arnpr_{\sigma,\gamma}^{v}(\Phi)$ for Eve, and $\vA(\Pi\backslash\Phi)(v) = \sup_{\gamma\in\adamstrategy}\inf_{\sigma\in\evestrategy}\arnpr_{\sigma,\gamma}^{v}(\Pi \backslash \Phi)$ for Adam.
A strategy $\sigma$ for Eve is \textit{optimal} from vertex $v$ if $ \inf_{\gamma\in\adamstrategy} \arnpr_{\sigma,\gamma}^{v}(\Phi) = \vE(\Phi)(v)$.
Optimal strategies for Adam are defined analogously. 

We divide solving stochastic games into three distinct tasks. Given an SG, 
solving the SG \textit{quantitatively} amounts to computing the values of all vertices in the arena. 
Solving the SG \emph{strategically} amounts to computing an optimal strategy of Eve (or Adam) for the game.

Since for both SSGs and SPGs, solving quantitatively and strategically is polynomially equivalent~\cite{andersson2009complexity}, we just say "solving" in what follows. We mainly consider quantitative solving, but Theorem~\ref{th:reduction} applies to both quantitative and strategic solving.

\paragraph*{Determinacy}
\textit{Determinacy} refers to the property of an SG where both players, Eve and Adam, have optimal strategies, meaning they can guarantee to achieve the values of the game, regardless of the strategies employed by the other player. \textit{Pure memoryless determinacy} means that both players have pure memoryless optimal strategies.
\begin{theorem}[Pure Memoryless Determinacy~\cite{martin1998determinacy}]\label{th:determinacy}

Let $(G,\Phi)$ be an SG, where $\Phi$ is a reachability or parity objective. For all $v\in V$, it holds that $\vE(\Phi)(v) + \vA(\Pi\backslash\Phi)(v) = 1$. 
Pure memoryless optimal strategies exist for both players from all vertices. 

\end{theorem}

When Eve and Adam follow pure memoryless strategies $\sigma\in\Sigma_{\exists}$ and $\gamma\in\Sigma_{\forall}$ respectively, we obtain sub-arena $G_{\sigma,\gamma}$, which can be seen as the MC $\mc_{\sigma,\gamma}$.
We can reduce the winning probabilities $\arnpr^{v_I}_{\sigma,\gamma}$ to reachability probabilities in $G_{\sigma,\gamma}$ as follows. 
Given a reachability objective $\greach(T)$, $\arnpr^{v_I}_{\sigma,\gamma}(\greach(T)) = \Pr^{v_I}_{\sigma,\gamma}(\Reach(T))$. 
Given a parity objective $\gparity(p)$, , i.e. $(G,\gparity(p))$, we call $B\in\textit{BSCC}(\mc_{\sigma,\gamma})$ an \textit{even BSCC} if $min(p(B))$ is even, meaning intuitively the smallest priority of its vertices is even. \textit{Odd BSCCs} are defined analogously. Then $\arnpr^{v_I}_{\sigma,\gamma}(\gparity(p)) = \Pr^{v_I}_{\sigma,\gamma}(\Reach(B_E))$, where 
    $B_E = \bigcup_{min(p(B))\textrm{ is even}} B\in \textit{BSCC}(\mc_{\sigma,\gamma})$.

\begin{corollary}[Sufficiency of Pure Memoryless Strategies~\cite{chatterjee2008reduction}]\label{corollary:pms}
Let $G=((V,E), (\Veve, \Vadam, V_{R}), \arntrans)$ be a stochastic arena, $\greach(T)$ a reachability objective, and $\gparity(p)$ a parity objective. For all vertices $v\in V$, it holds:
\begin{itemize}
    \item 

    $\vE(\greach(T))(v) = \displaystyle{\sup_{\sigma\in\evestrategy}\inf_{\gamma\in\adamstrategy}\arnpr_{\sigma,\gamma}^{v}(\greach(T)) 
             = \sup_{\sigma\in\Sigma_{\exists}}\inf_{\gamma\in\Sigma_{\forall}}\Pr_{\sigma,\gamma}^{v}(\Reach(T))}$
    \item 
    $\vE(\gparity(p))(v) = \displaystyle{\sup_{\sigma\in\evestrategy}\inf_{\gamma\in\adamstrategy}\arnpr_{\sigma,\gamma}^{v}(\gparity(p)) 
             = \sup_{\sigma\in\Sigma_{\exists}}\inf_{\gamma\in\Sigma_{\forall}}\Pr^{v}_{\sigma,\gamma}(\Reach(B_E))}$
    \\ where $B_E = \bigcup_{min(p(B))\textrm{ is even}} B\in \textit{BSCC}(\mc_{\sigma,\gamma})$.
\end{itemize}

\end{corollary}
Therefore we consider only pure memoryless strategies in the sequel, unless stated otherwise.

\section{A Gadget for Transforming SPGs into SSGs}\label{sec:gadegts}

The aim of this paper is to reduce an SPG $(G,\gparity(p))$ to an SSG $(\wG, \greach(\vwin))$ such that the probability of reaching target vertex $\vwin$ in this SSG is related to the probability of winning in the SPG $(G,\gparity(p))$. 
As an important step toward this goal, we introduce in this section a gadget that expands each transition of $G$ while removing the priority function.

Let $G=((V,E), (\Veve, \Vadam, V_{R}), \arntrans)$ be a stochastic arena, $p:V\rightarrow\mathbb{N}$ be a priority function, and $(G,\gparity(p))$ be an SPG. 
Section~\ref{subsec:gadgets} presents the gadget enabling the reduction 
from SPG to SSG $(\wG, \greach(\vwin))$. We then analyze how probabilistic events in $\wG$ are related to those in $G$. Section~\ref{subsec:crosspath} presents a bound on the probability of reaching BSCCs in $\wG$. Section~\ref{subsec:wineven} provides a bound on the winning probability once a BSCC in $\wG$ is reached, while Section~\ref{subsec:interval} gives interval bounds on the winning probabilities in $\wG$ with regard to those in $G$.

\subsection{Gadget Construction}\label{subsec:gadgets}

To reduce the parity objective to a reachability objective, we transform the SPG $(G,\gparity(p))$ into the SSG $(\wG,\greach(\vwin))$ by means of a gadget, whose structure was defined by Chatterjee and Fijalkow in~\cite{Chatterjee_2011}. The specific values they introduce give a reduction from deterministic parity games to SSGs, but does not work for a reduction from SPGs to SSGs because the probability they use are not small enough. 
The intuition of the gadget is as follows: whenever a play visits a vertex with even priority in $G$, give a small but positive chance to reach a winning sink in $\wG$. Vertices with odd priority yield a small chance to reach a losing sink. Finally, to represent that smaller priorities have precedence over larger ones, the probability of reaching a sink from a vertex depends on the priority it is associated to. We introduce a monotonically decreasing function $\alpha$ for this purpose. 

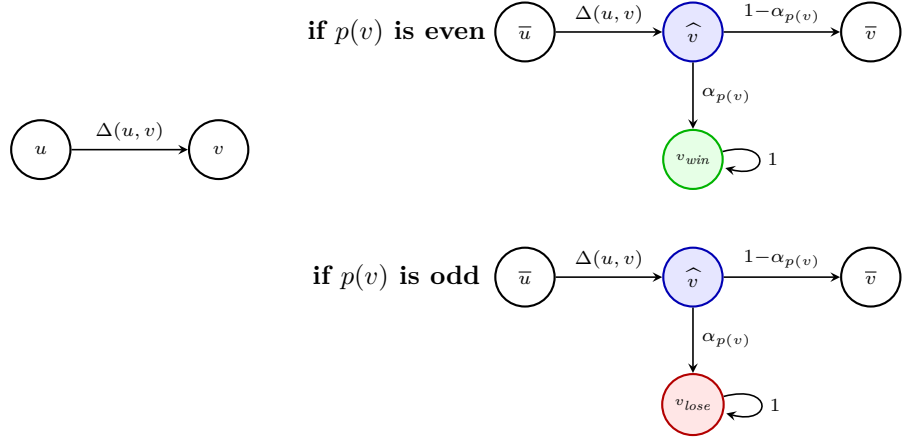
\begin{figure}
\centering
\begin{tikzpicture}[scale=0.13,
]
\node[plainNodes] (u) at (-4,-27) {$u$};
\node[plainNodes] (v) at (14,-27) {$v$};
\draw[->] (u) -- node[labelNodes, above]{$\arntrans(u,v)$} (v);

\node at (32,-15) {\textbf{if $p(v)$ is even}};
\node[plainNodes] (u_prime_1) at (45,-15) {$\ou$};
\node[hatNodes] (u_v_1) at (62,-15) {$\widehat{v}$};
\node[winNode] (win) at (62,-28) {$\vwin$};
\node[plainNodes] (v_prime_1) at (80,-15) {$\ov$};
\draw[->] (u_prime_1) -- node[labelNodes, above] {$\arntrans(u,v)$} (u_v_1);
\draw[->] (u_v_1) -- node[labelNodes, above] {$1{-}\alpha_{p(v)}$}(v_prime_1);
\draw[->] (u_v_1) -- node[labelNodes,right] {$\alpha_{p(v)}$} (win);
\draw[->] (win) edge [loop right] node[labelNodes] {$1$} ();

\node at (32,-40) {\textbf{if $p(v)$ is odd}};
\node[plainNodes] (u_prime_2) at (45,-40) {$\ou$};
\node[hatNodes] (u_v_2) at (62,-40) {$\widehat{v}$};
\node[loseNode] (lose) at (62,-53) {$\vlose$};
\node[plainNodes] (v_prime_2) at (80,-40) {$\ov$};
\draw[->] (u_prime_2) -- node[labelNodes, above] {$\arntrans(u,v)$} (u_v_2);
\draw[->] (u_v_2) -- node[labelNodes, above] {$1{-}\alpha_{p(v)}$} (v_prime_2);
\draw[->] (u_v_2) -- node[labelNodes,right] {$\alpha_{p(v)}$} (lose);
\draw[->] (lose) edge [loop right] node[labelNodes] {$1$} ();
\end{tikzpicture}
\caption{The gadgets for reducing SPG $G$ (left) to SSG $\wG$ (right)}
\label{fg:r2gadget}
\end{figure}

We obtain the stochastic arena $\wG$ by modifying $G$ as indicated in Figure~\ref{fg:r2gadget}. 
Each vertex $v$ in $G$ is duplicated in $\wG$ yielding vertices $\widehat{v}$ and $\ov$. A transition $\arntrans(u,v)$ in $G$ is replaced by first moving to $\widehat{v}$, which can then either evolve to a sink with probability $\alpha_{p(v)}$, or to the copy  $\ov$ with the complementary probability. Depending on $p$ being even or odd, the sink is $\vwin$ or $\vlose$.

Formally, for $U \subseteq V$, let $\oU = \{\ov\ |\ v\in U\}$, $\widehat{U}=\{\widehat{v}\ |\ v\in U\}$ and $\widetilde{U}= \oU \uplus \widehat{U}$.
We define the arena $\wG = ((\wV \uplus\{\vwin,\vlose\}, \wE), 
(\oVeve, \oVadam, \oV_R\uplus\widehat{V}\uplus\{\vwin,\vlose\} ), \warntrans)$
where the new edge set $\wE$ is as follows: $
\wE = \  \{(\ou,\widehat{v})\ |\ (u,v)\in E\} \ \uplus 
     \{ (\widehat{v}, \ov), (\widehat{v},\vwin)\ |\ v\in V,\ p(v)\ \text{is even}\} \ \uplus 
     \{ (\widehat{v}, \ov), (\widehat{v},\vlose)\ |\ v\in V,\ p(v)\ \text{is odd}\} \ \uplus 
     \{(\vwin,\vwin), (\vlose,\vlose)\}   $.

To define the new transition function $\warntrans$, let $\alpha:\mathbb{N}\rightarrow[0,1]$ where $\alpha_i$ represents the probability of entering the winning (resp. losing) sink before visiting a vertex with even (resp. odd) priority $i$. We give suitable values for $\alpha$ later, in Lemma~\ref{lm:ar-alpha} on page~\pageref{lm:ar-alpha}. 
Now, we define $\warntrans:\wV \times \wV \rightarrow [0,1]$ as follows:
\[
\warntrans(\wu,\ww) = 
\begin{cases}
\arntrans(u,w) & \text{if }\wu\in\oV, \ww\in\widehat{V} \\
1-\alpha_{p(u)} & \text{if }\wu\in\widehat{V}, \ww\in \oV, u=w \\
\alpha_{p(u)} & \text{if }\wu\in\widehat{V},p(u)\text{ is even, } \ww=\vwin \\
\alpha_{p(u)} & \text{if }\wu\in\widehat{V},p(u)\text{ is odd, } \ww=\vlose \\
1 & \text{if }\wu=\ww=\vwin\text{ or }\wu=\ww=\vlose \\
0 & \text{otherwise}
\end{cases}
\]
When the context is clear, we also address the SPG $(G,\gparity(p))$ and the SSG $(\wG,\greach(\vwin))$ with $G$ and $\wG$ respectively.

Since all new vertices $\widehat{V}\uplus\{\vwin,\vlose\}$ are random vertices, a strategy of either player in SPG $G$ is a strategy in SSG $\wG$ and vice versa. That is, there is a one-to-one relationship between strategies in $G$ and $\wG$. Hence, to keep notations simpler, we do not distinguish between strategies in $G$ and $\wG$. 

A pair of strategies $\sigma,\gamma\in\Sigma_{\exists}\times\Sigma_{\forall}$ for Eve and Adam in $G$ induces the sub-arena $G_{\sigma, \gamma}$. Similarly, we obtain $\wG_{\sigma, \gamma}$. 
If $U$ is an even or odd BSCC in SPG $G_{\sigma,\gamma}$, we denote with $\widetilde{U}$ what we call the associated \textit{even pBSCC} or \textit{odd pBSCC} in SSG $\wG_{\sigma,\gamma}$ respectively.  While those are not BSCCs, they correspond to the BSCC of the associated parity game, and we never consider the only true BSCCs of $\wG_{\sigma,\gamma}$, i.e. $\{\vwin\}$ and $\{\vlose\}$.

We continue with the example in Figure~\ref{fig:running-example-G} and Figure~\ref{fig:running-example-subG}. For the vertices \(v_0\), \(v_1\), and \(v_2\), we assign priority \(0\), and for each remaining vertex \(v_i\), $i\in\{3,4,5\}$, we assign priority \(i\). An illustration of the corresponding sub-arena \(\widetilde{G}_{\sigma,\gamma}\), induced by our gadget construction, is provided in Figure~\ref{fig:running-example-subtildeG}. 
Since we do not need to distinguish between different types of vertices, we use circles for all vertices.
In \(G_{\sigma,\gamma}\), the set \(\{v_3, v_4, v_5\}\) forms an odd BSCC. We refer to the set \(\{\widehat{v}_3, \overline{v}_3, \widehat{v}_4, \overline{v}_4, \widehat{v}_5, \overline{v}_5\}\) in \(\widetilde{G}_{\sigma,\gamma}\) as the associated odd pBSCCn with \(\widehat{v}_3\), \(\widehat{v}_4\), and \(\widehat{v}_5\) having outgoing transitions to either \(\vwin\) or \(\vlose\).

\begin{figure}
\centering
\begin{tikzpicture}[
    scale=1.8, 
]
    \node[hatNodes] (00) at (-1,0) {$\widehat{v}_0$};
    \node[plainNodes] (0) at (0,0) {$\ov_0$};
    \node[hatNodes] (01) at (1,0) {$\widehat{v}_1$};
    \node[plainNodes] (1) at (2,0) {$\ov_1$};
    \node[hatNodes] (02) at (3,0) {$\widehat{v}_2$};
    \node[plainNodes] (2) at (4,0) {$\ov_2$};
    \node[hatNodes] (03) at (5,0) {$\widehat{v}_4$};
    \node[plainNodes] (3) at (6,0) {$\ov_4$};

    \node[hatNodes] (05) at (6,2) {$\widehat{v}_5$};
    \node[plainNodes] (5) at (5,2) {$\ov_5$};
    \node[hatNodes] (04) at (4,2) {$\widehat{v}_3$};
    \node[plainNodes] (4) at (3,2) {$\ov_3$};

    \node[winNode] (7) at (1,-0.8) {$\vwin$};
    \node[loseNode] (8) at (4,1.3) {$\vlose$};
    
    \draw[->] (0,-0.5) -- (0);
    \draw[->] (00) -- node[tinylabelNodes,below] {$1{-}\alpha_0$} (0);
    \draw[->] (0) -- node[tinylabelNodes,below] {$0.1$} (01);
    \draw[->] (01) -- node[tinylabelNodes,below] {$1{-}\alpha_0$} (1);
    \draw[->] (1) -- node[tinylabelNodes,below] {$0.1$} (02);
    \draw[->] (02) -- node[tinylabelNodes,below] {$1{-}\alpha_0$} (2);
    \draw[->] (2) -- node[tinylabelNodes,below] {$0.1$} (03);
    \draw[->] (03) -- node[tinylabelNodes, below] {$1{-}\alpha_4$} (3);
    \draw[->] (3.east) -- ++(0.5,0) |- node[tinylabelNodes, right, pos=0.3] {$1$} (05.east);
    
    \draw[->] (05) -- node[tinylabelNodes, above] {$1{-}\alpha_5$} (5);
    \draw[->] (5) -- node[tinylabelNodes, above] {$0.5$} (04);
    \draw[->] (04) -- node[tinylabelNodes, above] {$1{-}\alpha_3$} (4);
    \draw[->] (4) -- ++(0,0.5) -| node[tinylabelNodes, above, pos=0.3] {$1$} (05);
    \draw[->] (5) -- node[tinylabelNodes, right, pos=0.8] {$0.5$} (03);
    \draw[->] (04) -- node[tinylabelNodes, left] {$\alpha_3$} (8);
    \draw[->] (05) |- node[tinylabelNodes, below, pos=0.9] {$\alpha_5$} (8);
    
    \draw[->] (0) -- ++(0,0.3) -| node[tinylabelNodes,above, pos=0.2]{$0.9$} (00.north east);
    \draw[->] (1) -- ++(0,0.55) -| node[tinylabelNodes,above, pos=0.35]{$0.9$} (00);
    \draw[->] (2) -- ++(0,0.8) -| node[tinylabelNodes,above, pos=0.35]{$0.9$} (00.north west);
    \draw[->] (00) |- node[tinylabelNodes, above, pos=0.9]{$\alpha_0$} (7.west);
    \draw[->] (01) -- node[tinylabelNodes, right]{$\alpha_0$} (7.north);
    \draw[->] (02) |- node[tinylabelNodes, above, pos=0.8]{$\alpha_0$} (7.north east);
    \draw[->] (03) |- node[tinylabelNodes, above, pos=0.8]{$\alpha_4$} (7.south east);
    \end{tikzpicture}
    \caption{The sub-arena $\wG_{\sigma,\gamma}$ induced by the gadget}
    \label{fig:running-example-subtildeG}
\end{figure}
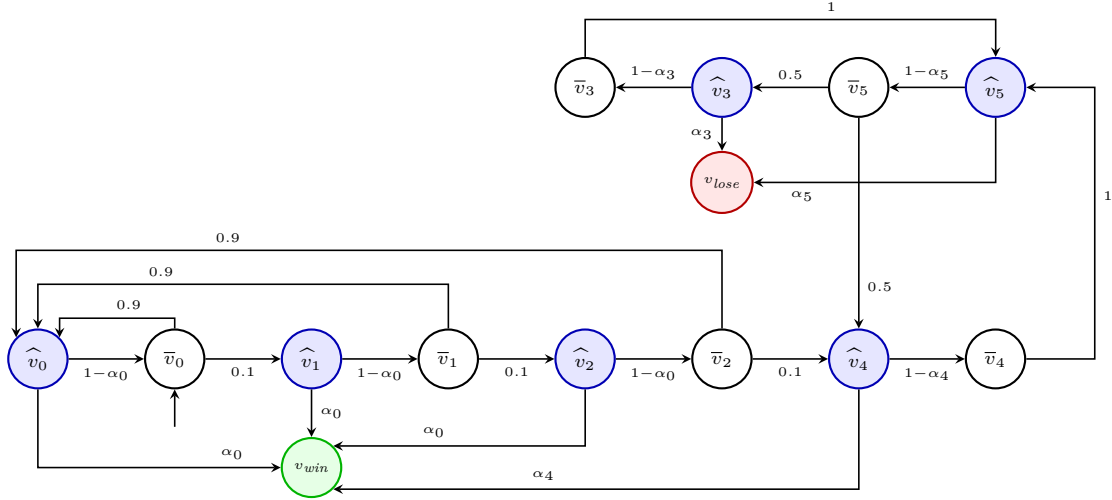

\subsection{Before Entering a pBSCC in SSG $\wG$}
\label{subsec:crosspath}

Recall that when Eve and Adam follow pure memoryless strategies $\sigma \in \Sigma_{\exists}$ and $\gamma \in \Sigma_{\forall}$, the resulting sub-arenas $G_{\sigma,\gamma}$ and $\wG_{\sigma,\gamma}$ can be viewed as finite Markov chains. We first focus on what happens before a play reaches a pBSCC in $\wG_{\sigma,\gamma}$. 
Specifically, we give a lower bound on the probability of reaching an entry state of a pBSCC without entering a winning or losing sink.
Later, in Lemma~\ref{lm:ar-alpha}, we use this bound to determine a suitable value for \(\alpha_0\).

Intuitively, we show that the probability of entry is minimized in a classical worst-case scenario extending the sub-arena \(\{v_0, v_1, v_2\}\) in Figure~\ref{fig:running-example-subG}, and \(\{\widehat{v}_0, \overline{v}_0, \widehat{v}_1, \overline{v}_1, \widehat{v}_2, \overline{v}_2\}\) in Figure~\ref{fig:running-example-subtildeG}. 
More precisely, we consider an original sub-arena \(G_{\sigma,\gamma}\), with \(n\) states arranged in a sequence (as \(\{v_0, v_1, v_2\}\) in Figure~\ref{fig:running-example-subG}) before reaching a BSCC. Each state has a minimal probability \(\pM\) of progressing to the next state and a maximal probability \(1 - \delta_{\min}\) of returning to the initial state. All these states are assigned parity value \(0\). Upon applying our gadget construction, we introduce corresponding random states (\(\{\widehat{v}_0, \widehat{v}_1, \widehat{v}_2\}\) in Figure~\ref{fig:running-example-subtildeG}) that have the highest probability \(\alpha_0\) of transitioning to the winning sink \(\vwin\).

In the following, let $(G,\gparity(p))$ be an SPG, and $(\wG,\greach(\vwin))$ be its associated SSG. For all strategy pairs $\sigma,\gamma\in\Sigma_{\exists}\times\Sigma_{\forall}$, let $\wPr_{\sigma,\gamma}^{\ov}(\crossPath)$ denote the probability for a play starting from $\ov\in\oV$ to reach a pBSCC in $\wG$. 
Note that we never consider $\widehat{v} \in \widehat{V}$ as the starting vertex.

\begin{restatable}{lemma}{lowercross}\label{lm:formal-crosspath}
 For all strategy pairs $\sigma,\gamma\in\Sigma_{\exists}\times\Sigma_{\forall}$, for all $v\in V$, it holds:
\begin{equation*}
    \wPr^{\ov}_{\sigma,\gamma}(\crossPath) \ge \frac{(1-x_0)x_0^n}{(1-x_0)-(1-x_0^n)x_1}
\end{equation*}
where $n = |V|$, $x_0 = \pM(1-\alpha_0)$, $x_1 =(1-\pM)(1-\alpha_0)$, and $\pM = \min\limits_{u,v\in V}\{\arntrans(u,v) \mid \arntrans(u,v)>0\}$.
\end{restatable}

\textbf{Sketch of Proof:}
We fix an arbitrary strategy pair $\sigma,\gamma\in\Sigma_{\exists}\times\Sigma_{\forall}$, and analyze the corresponding MC~$\wG_{\sigma,\gamma}$. We simplify the MC while either preserving or under-approximating the probability of reaching a pBSCC in $\wG_{\sigma,\gamma}$. These steps merge all pBSCCs into a sink $v_b$, eliminate auxiliary states to simplify the MC, 
increase all values of $\alpha$, and restructure transitions so that only one designated vertex can reach the sink $v_b$ directly. We denote the resulting MC with $\wG_4$.
We then derive a lower bound on the probability of reaching $v_b$ in $\wG_4$, 
which provides a reachability lower bound in a template MC with absorbing sinks and bounded transition probabilities. 
As the reachability probabilities of $v_b$ in $\wG_4$ underapproximate those in $\wG_{\sigma,\gamma}$, this yields the desired lower bound on $\wPr^{\ov}_{\sigma,\gamma}(\crossPath)$. 
The full proof of this lemma can be found in Appendix~\ref{app-subsec:crosspath}.

\subsection{Inside a pBSCC in SSG \wG}\label{subsec:wineven}
We now focus on what happens after a play reaches a pBSCC in sub-arena $\wG_{\sigma,\gamma}$. Specifically, we give a lower bound on the probability of reaching the winning sink after reaching an even pBSCC, and dually an upper bound on the probability of reaching the winning sink after reaching an odd pBSCC.

The lower bound is attained in the MC shown in Figure~\ref{fig:achievable2}, where $k$ is an even parity value.
There are $2n+1$ states in a line, and winning and losing sinks. Each white state has maximal probability $1-\delta_{\min}$ to return to the initial state, and otherwise proceeds to the next blue state.
Each blue state, except $\widehat{v}'$, can with probability $\aKP$ go to the losing sink, and otherwise proceeds to the next white state. The special state $\widehat{v}'$ goes with probability $\aK$ to $\vwin$, and otherwise proceeds to $\ov$.
Unlike the case with \(\{\widehat{v}_0, \overline{v}_0, \widehat{v}_1, \overline{v}_1, \widehat{v}_2, \overline{v}_2\}\) in Figure~\ref{fig:running-example-subtildeG}, this MC cannot be obtained by applying our gadget on some sub-arena $G_{\sigma,\gamma}$, and hence this bound is not guaranteed to be tight. More precisely, the outgoing transitions of $\widehat{v}$ indicate that $v$ has an odd parity value, while $\widehat{v}'$ suggests otherwise. 
The upper bound is obtained by considering the same MC, where $k$ is an odd parity value.
Later, in Lemma~\ref{lm:ar-alpha}, we use these two bounds to find suitable values for all $\alpha_k$ with $k\in\mathbb{N}$.

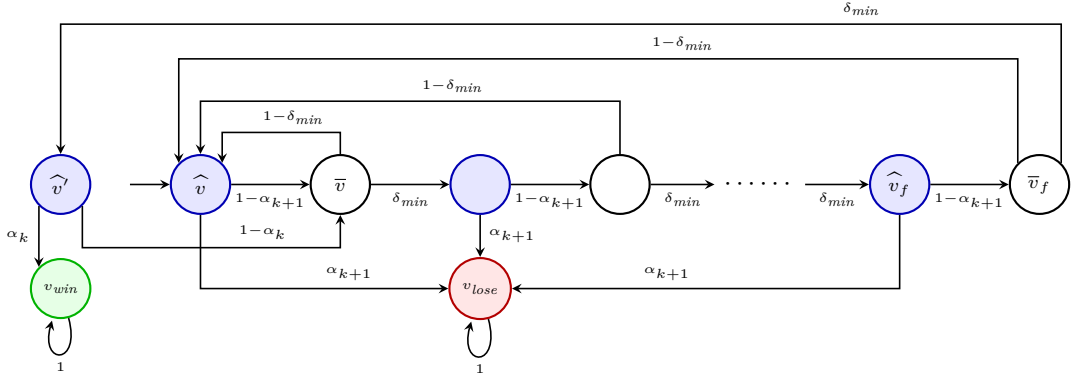
\begin{figure}[]
    \centering
    \begin{tikzpicture}[scale=0.92]
    \node[circle] (-1) at (-3.2,0) {};
    \node[hatNodes] (00) at (-2,0) {$\widehat{v}$};
    \node[plainNodes] (0) at (0,0) {$\ov$};
    \node[hatNodes] (01) at (2,0) {};
    \node[plainNodes] (1) at (4,0) {}; 
    \node[rectangle] (5) at (6,0) {$\cdots\cdots$};
    \node[hatNodes] (06) at (8,0) {$\widehat{v}_f$};
    \node[plainNodes] (6) at (10,0) {$\ov_f$};
    \node[hatNodes] (add) at(-4,0) {$\widehat{v}'$};
    \node[winNode] (8) at (-4,-1.5) {$\vwin$};
    \node[loseNode] (7) at (2,-1.5) {$\vlose$};

    \draw[->] (-1) -- (00);
    \draw[->] (00) -- node[labelNodes,below,font=\tiny] {$1{-}\aKP$} (0);
    \draw[->] (0) -- node[labelNodes,below,font=\tiny] {$\pM$} (01);
    \draw[->] (01) -- node[labelNodes,below, font=\tiny] {$1{-}\aKP$} (1);
    \draw[->] (1) -- node[labelNodes,below,font=\tiny] {$\pM$} (5);
    \draw[->] (5) -- node[labelNodes,below,font=\tiny] {$\pM$} (06);
    \draw[->] (06) -- node[labelNodes,below, font=\tiny] {$1{-}\aKP$} (6);

    
    \draw[->] (00) |- node[labelNodes,above, pos=0.8, font=\tiny] {$\aKP$} (7);
    \draw[->] (01) -- node[labelNodes,right, pos=0.5, font=\tiny] {$\aKP$} (7);
    \draw[->] (06) |- node[labelNodes, above,pos=0.8, font=\tiny]  {$\aKP$} (7.east);

    \draw[->] (0) -- ++(0,0.75) -| node[labelNodes, above,pos=0.2, font=\tiny] {$1{-}\pM$} (00.north east);
    \draw[->] (1) -- ++(0,1.2) -| node[labelNodes, above,pos=0.2, font=\tiny] {$1{-}\pM$} (00);
    \draw[->] (6.north west) -- ++(0,1.5) -| node[labelNodes,above,pos=0.2, font=\tiny] {$1{-}\pM$} (00.north west);
    \draw[->] (6.north east) -- ++(0, 2) -| node[labelNodes, above, pos=0.1, font=\tiny] {$\pM$} (add);
    \draw[->] (add.south east) -- ++(0,-0.6) -| node[above, pos=0.35, labelNodes, font=\tiny]{$1{-}\aK$} (0);
    \draw[->] (add.south west) -- node[labelNodes, left, font=\tiny] {$\aK$} (8.north west);
    \draw (8) edge[loop below] node[labelNodes, font=\tiny] {$1$} (8);
    \draw (7) edge[loop below] node[labelNodes, font=\tiny] {$1$} (7);
    \end{tikzpicture}
    \caption{The MC where the lower bound on minimum $\winEven$ probability is attained}
    \label{fig:achievable2}
\end{figure}

Let $U$ be an even BSCC with smallest priority $k$ in $G_{\sigma,\gamma}$ and $\wU$ its associated pBSCC in $\wG_{\sigma,\gamma}$. Let $\wPr_{\sigma,\gamma}^{k,\wU}(\winEven)$ denote the minimum probability of reaching the winning sink after reaching $\wU$. That is, $\wPr_{\sigma,\gamma}^{k,\wU}(\winEven) = \min\{\wPr_{\sigma,\gamma}^{\wv}(\Reach(\vwin)) \mid \wv\in\wU\}$. We denote it $\wPr_{\sigma,\gamma}^{k,\wU}(\winEven)$ when $\wU$ is clear from context. Analogously, given an odd BSCC $U$ with smallest priority $k$ in $G_{\sigma,\gamma}$, we use $\wPr^{k}_{\sigma,\gamma}(\winOdd)$ to denote the maximum probability of reaching the winning sink after reaching $\wU$.

\begin{restatable}{lemma}{lowerwineven}\label{lm:formal-wineven}
For all strategy pairs $\sigma,\gamma\in\Sigma_{\exists}\times\Sigma_{\forall}$, for all even $k$, it holds:
\begin{equation}
\wPr_{\sigma,\gamma}^{k}(\winEven) \ge (1-\alpha_{k+1}) \cdot \frac{(1-x_2)\cdot x_2^{n-1}\cdot x_4}{1-(x_2+x_3) + x_5\cdot x_2^{n} + t\cdot x_2^{n-1} - x_5 \cdot x_2^{n-1}} \notag
\end{equation}
and for all odd $k$, it holds:
\begin{equation}
\wPr_{\sigma,\gamma}^{k}(\winOdd) \leq 1 - (1-\alpha_{k+1}) \cdot \frac{(1-x_2)\cdot x_2^{n-1}\cdot x_4}{1-(x_2+x_3) + x_5\cdot x_2^{n} + t\cdot x_2^{n-1} - x_5 \cdot x_2^{n-1}}\notag
\end{equation}
where  $n = |V|$, $x_2 = \pM(1-\alpha_{k+1})$, $x_3 = (1-\pM)(1-\alpha_{k+1})$, $x_4 = \pM\alpha_{k}$, $x_5 = \pM(1-\alpha_{k})+(1-\pM)(1-\alpha_{k+1})$, and $\pM$ is as before. 
\end{restatable}
\textbf{Sketch of Proof:}
The proof follows the same structure as the one for Lemma~\ref{lm:formal-crosspath}, applied to a BSCC in $\mc_{\sigma,\gamma}$. 
 We apply a similar four-step transformation and obtain a simplified MC where 
we can directly derive an upper bound. The lower bound comes as the dual. 
he full proof of this lemma can be found in Appendix~\ref{app-subsec:wineven}.

\subsection{Range of Winning Probabilities in the SSG}\label{subsec:interval}

We now relate the winning probabilities in the constructed SSG $\wG$ to the original SPG $G$. 
Intuitively, with a fixed strategy pair $\sigma,\gamma\in\Sigma_{\exists}\times\Sigma_{\forall}$, the value $\widetilde{\arnpr}_{\sigma,\gamma}$ of the SSG $\wG$ falls into a range around the value $\arnpr_{\sigma,\gamma}$ of the SPG $G$, and the range size depends on the probabilities $\crossPath$, $\winEven$ and $\winOdd$. 
\begin{restatable}{lemma}{interval}\label{lm:interval}
Let $x,y\in(0,1)$ such that for all even $k$ $\wPr_{\sigma,\gamma}^\ov(\crossPath) > x$ and $\wPr_{\sigma,\gamma}^{k}(\winEven) \ge y$, and for all odd $k$, $\wPr_{\sigma,\gamma}^{k}(\winOdd)\leq 1-y$, then it holds:
\[
y\cdot \arnpr^{v}_{\sigma,\gamma} - y+x\cdot y \  \leq \
\oarnpr^{\ov}_{\sigma,\gamma} \ \leq \
\arnpr^{v}_{\sigma,\gamma} + 1-x \cdot y
\]    
\end{restatable}
The proof of this lemma is quite calculatory, and can be found in Appendix~\ref{app-subsec:interval}.

\section{Reducing SPGs to SSGs}\label{sec:reduction}

We now present the direct reduction from SPGs to SSGs. 
Let $G=((V,E), (\Veve, \Vadam, V_{R}), \arntrans)$ be a stochastic arena, $p:V\rightarrow\mathbb{N}$ be a priority function, and $(G,\gparity(p))$ be an SPG. We construct the SSG $(\wG,\greach(\vwin))$ using the gadget presented in Section~\ref{subsec:gadgets}. 
Section~\ref{subsec:interval-bounds} presents a lower bound on the difference between winning probabilities associated to different strategy pairs in $G$. Section~\ref{subsec:reduction-theorem} presents the main theorem establishing the reduction, while Section~\ref{subsec:complexity} gives complexity bounds.

\subsection{A Lower Bound on Different Strategies}\label{subsec:interval-bounds}

We consider two strategy pairs $(\sigma,\gamma),(\sigma',\gamma')\in\Sigma_{\exists}\times\Sigma_{\forall}$ and show a general result on all such pair: if they yield different values in $G$, then there exists a lower bound on the difference between these values.

In the following, we assume for all $u\in V_R, v\in V$ that $\arntrans(u,v)$ is a rational number $\frac{a_{u,v}}{b_{u,v}}$, where $a_{u,v}\in\mathbb{N},b_{u,v}\in\mathbb{N}_{+}$ and $a_{u,v}\leq b_{u,v}$. Let $M = \max\limits_{(u,v)\in E}\{b_{u,v}\}$ and $n=|V|$.
\begin{restatable}{lemma}{lmepsilon}\label{lm:epsilon}
For all $(\sigma,\gamma), (\sigma',\gamma')\in\Sigma_{\exists}\times\Sigma_{\forall}$, for all $v \in V$, 
the following holds: 
\[
\arnpr^{v}_{\sigma,\gamma} \neq  \arnpr^{v}_{\sigma',\gamma'} \Rightarrow
|\arnpr^{v}_{\sigma,\gamma} -  \arnpr^{v}_{\sigma',\gamma'}| > \frac{1}{(n!)^2 M^{2n^2}} = \epsilon
\]
\end{restatable}
\begin{proof}

Let $\Pr^{v}_{\sigma,\gamma}(\enterEven)$ be the probability for a play starting from $v\in V$ to reach an even BSCC.
It follows from Corollary~\ref{corollary:pms} that for all $\sigma,\gamma\in\Sigma_{\exists}\times\Sigma_{\forall}$ and $v\in V$, we have $\arnpr^{v}_{\sigma,\gamma} = \Pr^{v}_{\sigma,\gamma}(\enterEven)$. 
We can obtain $\Pr^{v}_{\sigma,\gamma}(\enterEven)$ by setting all vertices belonging to at least one even BSCC as the target set, and calculating the reachability probability. Calculating~$\arnpr^{v}_{\sigma,\gamma}$ is thus reduced to solving a linear equation system $x=Ax+b$ according to Theorem~\ref{th:reachability}. 
We omit the details of $A$ and $b$. Every non-zero entry of $A$ and $b$ is either $1$, or $\frac{a_{u,v}}{b_{u,v}}$ for some $u,v\in V_R\times V$.

We use the following notations:
\newcommand{\qij}{\ensuremath{Q[i,j]}}
\newcommand{\cij}{\ensuremath{c_{i,j}}}
\newcommand{\dij}{\ensuremath{d_{i,j}}}
\newcommand{\bi}{\ensuremath{b[i]}}
\newcommand{\qi}{\ensuremath{Q[i]}}
\begin{itemize}
\item Let $s=|b|$. It follows that $s<n$ since there is at least one vertex in a BSCC. 
\item Let $Q=I-A$. For $i\in 1,2,\ldots,n$ we denote the $i$-th row of $Q$ with $\qi$, and the entry of $Q$ at $i$-th row and $j$-th column with $\qij$. It can be written as $\qij=\frac{\cij}{\dij}$, where $|\cij|$ and $|\dij|$ are natural numbers bounded by $M$ with $|\cij| \le |\dij|$. 
\item We denote the $i$-th entry of $b$ with $\bi$. It can be written as $\bi=\frac{c_{i,s+1}}{d_{i,s+1}}$, where $|c_{i,s+1}|$ and $|d_{i,s+1}|$ are natural numbers bounded by $M$ with $|c_{i,s+1}|\le |d_{i,s+1}|$.     
\end{itemize}
The equation system can be written as:
\[ Qx = b\]
We take an arbitrary row $i$, and write the $i$-th equation $\qi \cdot x = \bi$ as follows:
\begin{equation}\label{eq:before-multiply}
\begin{bmatrix} \frac{c_{i,1}}{d_{i,1}} & \frac{c_{i,2}}{d_{i,2}} & \cdots & \frac{c_{i,s}}{d_{i,s}} \end{bmatrix} \cdot x = \begin{matrix} \frac{c_{i,s+1}}{d_{i,s+1}} \end{matrix}
\end{equation}
We multiply equation~\eqref{eq:before-multiply} with $\prod^{s+1}_{t=1} d_{i,t}$ to obtain:
\begin{itemize}
    \item For all $j=1,\ldots,s$, $\qij$ equals $(\prod^{s+1}_{t=1} d_{i,t})\frac{\cij}{\dij}$, an integer with absolute value bounded by $M^{s+1}$.
    \item For all $i=1,\ldots,s$, $\bi$ equals $(\prod^{s}_{t=1} d_{i,t})c_{i,s+1}$, an integer with absolute value bounded by $M^{s+1}$.
\end{itemize}
We apply this transformation to each row of the equation system, and write the new equation system as:
\[ Q'x=b' \]
By Cramer's rule, for all $i=1,2,\ldots,s$, we obtain:
\[ x[i] = \frac{det(Q'_i)}{det(Q')} \]
where $Q'_i$ is the matrix obtained by replacing the $i$-th column of $Q'$ with the column vector~$b'$. It follows that all entries of $Q'_i$ are also integers with absolute values bounded by $M^{s+1}$.

Since $x[i]$ is a reachability probability, we have $x[i] \le 1$. Following from the calculation of determinants, we obtain the following: 
\[ |det(Q'_i)| \le |det(Q')| \le s!(M^{s+1})^s < n!M^{n^2} \]
Therefore, if the equation system resulting from $\sigma',\gamma'$ yields $x'[i]>x[i]$, we have: 
\[ x'[i]-x[i] > \frac{1}{(n!M^{n^2})^2} = \frac{1}{(n!)^2 M^{2n^2}} \qedhere \]  
\end{proof}

\subsection{Direct Reduction}\label{subsec:reduction-theorem}
We now establish the direct reduction from SPGs to SSGs.
\newcommand{\pvsg}{\ensuremath{\arnpr^{v}_{\sigma,\gamma}}}
\newcommand{\pvspg}{\ensuremath{\arnpr^{v}_{\sigma',\gamma}}}
\newcommand{\opovsg}{\ensuremath{\oarnpr^{\ov}_{\sigma,\gamma}}}
\newcommand{\wpovsg}{\ensuremath{\widetilde{\arnpr}^{\ov}_{\sigma,\gamma}}}
\newcommand{\opovspg}{\ensuremath{\oarnpr^{\ov}_{\sigma',\gamma}}}
\newcommand{\wpovspg}{\ensuremath{\widetilde{\arnpr}^{\ov}_{\sigma',\gamma}}}
\begin{theorem}[Reducing SPGs to SSGs]\label{th:reduction}
If for all $(\sigma,\gamma)\in\Sigma_{\exists}\times\Sigma_{\forall}$, and $v\in V$, the following conditions hold:
\begin{enumerate}
    \item\label{aspt:c0} $\wPr_{\sigma,\gamma}^{\ov}(\crossPath) >  \frac{4-\epsilon}{4}$ 
    \item\label{aspt:c1} $\wPr_{\sigma,\gamma}^{k}(\winEven) \ge \frac{4}{4+\epsilon}$ for all even
 $k$, and $\wPr_{\sigma,\gamma}^{k}(\winOdd) \le 1-\frac{4}{4+\epsilon}$ for all odd $k$
\end{enumerate}
where $\epsilon = \frac{1}{(n!)^2 M^{2n^2}}$,
then every optimal strategy $\sigma\in\Sigma_{\exists}$ of Eve in the SSG $(\wG,\greach(\vwin))$ is also optimal in the SPG $(G,\gparity(p))$.
The same holds for Adam.
\end{theorem}

\begin{proof}
We assume conditions~\ref{aspt:c0} and~\ref{aspt:c1} hold. We show that every optimal strategy $\sigma\in\Sigma_{\exists}$ in SSG $\wG$ is also optimal in SPG $G$. We prove this by contraposition.

We take $\Sigma_{\exists}^*\subseteq\Sigma_{\exists}$ and $\Sigma_{\forall}^*\subseteq\Sigma_{\forall}$ as the sets of optimal strategies of Eve and Adam in $\wG$. 
We obtain by Lemma~\ref{lm:interval} that for all $v \in V$ and all $\sigma,\gamma\in\Sigma_{\exists}^*\times\Sigma_{\forall}^*$, the following holds:
\[
y \cdot \pvsg - y + x \cdot y \ \leq \ \wpovsg \ \leq \ \pvsg + 1-x\cdot y
\]
Since conditions~\ref{aspt:c0} and~\ref{aspt:c1} hold, we substitute $x$ and $y$ to obtain:
\begin{equation}\label{eq-qn:range-0}
\wpovsg \le \pvsg + \frac{2\epsilon}{4+\epsilon}    
\end{equation}
If 
$\sigma$ is not optimal in $G$, then there exists another strategy $\sigma' \in \Sigma_{\exists}$ and a vertex $v\in V$ such that $\pvspg>\pvsg$. It follows again from Lemma~\ref{lm:interval} that:
\begin{equation}\label{eq-qn:range}
y \pvspg - y + x\cdot y \ \leq \ \wpovspg \ \leq \ \pvspg + 1-x\cdot y
\end{equation}
Furthermore, Lemma~\ref{lm:epsilon} yields:
\begin{equation}\label{eq-qn:range-1}
\pvspg>\pvsg+\epsilon    
\end{equation}
As a result, we obtain the following:
\begin{align*}
  \wpovspg \ge\ & y \cdot \pvspg - y+x\cdot y \tag*{\textit{by~\eqref{eq-qn:range}}} \\
>\ & y\cdot (\pvsg+\epsilon)-y+x\cdot y \tag*{\textit{by~\eqref{eq-qn:range-1}}} \\
=\ & \frac{4}{4+\epsilon} (\pvsg+\epsilon) - \frac{4}{4+\epsilon} \cdot \frac{\epsilon}{4} \\
=\ & \pvsg - \frac{\epsilon}{4+\epsilon}\pvsg+\frac{3\epsilon}{4+\epsilon} \\
\ge\ & \pvsg + \frac{2\epsilon}{4+\epsilon} \\
\ge\ & \wpovsg \tag*{\textit{by~\eqref{eq-qn:range-0}}}
\end{align*}
It indicates that $\wpovspg > \wpovsg$, which contradicts the assumption that $\sigma\in\Sigma_{\exists}^*$.
\end{proof}

Until now, we have used the function $\alpha$ in Theorem~\ref{th:reduction}, obtaining inequalities relating parity values in SPG $G$ to transition probabilities in SSG $\wG$. We now give requirements for $\alpha$ that satisfy all these inequalities. 
\begin{restatable}{lemma}{aralpha}\label{lm:ar-alpha}
When the values of $\alpha$ are arranged as follows, the conditions in Theorem~\ref{th:reduction} are satisfied:
\begin{enumerate} 
\item If $\alpha_0\le\frac{\pM^n}{8\epsi}$, then condition~\ref{aspt:c0} is satisfied.
\item If for all $k\in\mathbb{N}$, the following holds, then condition~\ref{aspt:c1} is satisfied:
\begin{equation}
\frac{\alpha_{k+1}}{\alpha_k} \le \frac{\pM^n(1-\pM)}{8\epsi+1} \notag.
\end{equation}

\end{enumerate}
\end{restatable}

\textbf{Sketch of Proof:} Both cases follow a similar structure. For $\alpha_0$ (respectively ${\alpha_{k+1}}/{\alpha_k}$), we derive from the bound given by Lemma~\ref{lm:formal-crosspath} (resp. Lemma~\ref{lm:formal-wineven}) a corollary giving a bound that explicitly makes use of $\alpha$. We then directly obtain the two cases of Lemma~\ref{lm:ar-alpha} from these two bounds. The full proof of this lemma, detailing how to compute function $\alpha$ can be found in Appendix~\ref{app-subsec:ar-alpha}.

\subsection{Complexity Considerations}\label{subsec:complexity}

To introduce complexity results, we first define size of a stochastic game $G$ as
$|G| = |V| + |E| + |\arntrans|$ where $|\arntrans|$ is the space needed to store the transition function (which may be stored in unary or binary). 
A now longstanding result shows that most stochastic game settings are polynomially reducible one to the other. In particular:
\begin{theorem}[From Theorem 1 in~\cite{andersson2009complexity}]
Solving stochastic parity games and solving simple stochastic games is polynomial-time equivalent.
Either can be using unary or binary encoding. 
\end{theorem}

We show that the reduction we have introduced in this paper is polynomial with binary encoding. 
We recall that $M = \max\limits_{(u,v)\in E}\{b_{u,v}\}$ and $n=|V|$. 
\begin{theorem}\label{th:size}
    Given an SPG $G$, there exist polynomial values for function $\alpha$ that satisfy Theorem~\ref{th:reduction}, such that the SSG $\wG$ is of size $\bigO(n^5 \log M )$ in binary.
\end{theorem}
\begin{proof}
Since $\pM \ge \frac{1}{M}$ and $1-\pM\ge\frac{1}{2}$, the following is a valid instance of $\alpha$, polynomial in $G$ (and polynomial in the transition probabilities appearing in $G$) under binary encoding: 
\[
\forall k\in\mathbb{N}, \ \alpha_k = \left( \frac{1}{16 (n!)^2 M^{2n^2+n} + 1} \right)^{k+1}
\]
Then the size of the SSG $\wG$ is:
\begin{align*}
|\wG| =\ & |\wV \uplus\{\vwin,\vlose\}| + |\wE| + |\warntrans| \\
= \ & \bigO(n) + \bigO(n^2) + \bigO(n^2)\cdot\bigO(n\cdot (n \log n + n^2\log M)) \\
= \ & \bigO(n^5 \log M )
\qedhere
\end{align*}
\end{proof}
According to~\cite{andersson2009complexity}, quantitative SSGs under unary and binary encoding are in the same complexity class, and so in $\textbf{NP}\,\mathbf{\cap}\,\textbf{coNP}$~\cite{condon1992complexity}. We thus obtain that our reduction yields an $\textbf{NP}\,\mathbf{\cap}\,\textbf{coNP}$ algorithm for solving SPGs.

\section{Epilogue}\label{ch:conclusion}
We have given a polynomial reduction from quantitative SPGs to quantitative SSGs, taking inspiration from a gadget used in~\cite{Chatterjee_2011} to obtain a reduction from deterministic PGs to quantitative SSGs.
After fixing a pair of strategies, the values of both the SPG and the SSG are determined, but the construction of the SSG makes it difficult to establish coinciding values. 
Using these fixed strategies, we showed that the value of the SSG falls into a range around the value of the SPG, where this range depends on the probability to reach a pBSCC of the SSG and the minimum probability to reach a winning sink in pBSCCs of the SSG. 
When considering all possible strategy pairs, we obtained a lower bound $\epsilon$ on their value differences in the SPG, by restricting transition probabilities to rational numbers and analyzing reachability equation systems of Markov chains. 
We then showed that by arranging transition probabilities of the SSG properly in terms of the size of the SPG, its smallest probability, and $\epsilon$, the value ranges of different strategy pairs can be narrowed so that they do not overlap. In this case, a reduction from SPGs to SSGs is achieved. 

Although under unary encoding exponential numbers can be introduced into the probability function $\alpha$ of the newly constructed SSGs, both reductions are polynomial. Hence, our construction yields an $\textbf{NP}\,\mathbf{\cap}\,\textbf{coNP}$ algorithm in both qualitative and quantitative SSGs under unary and binary encoding, substantiating the complexity results from previous works~\cite{condon1992complexity,andersson2009complexity,Chatterjee_2011}. .

Our result enables solving SPGs by first reducing them to SSGs and then applying algorithms for SSGs. However, its implementability is in question, due to the possibly huge representation of $\alpha$.  
Our reduction also captures the transformation from an MDP with a parity objective into an SSG. As we assume the minimum transition probability to be $\pM\in(0,\frac{1}{2}]$ in the original SPG, we cannot capture the subcase of reducing DPGs to quantitative SSGs. 
Although our reduction is unlikely to be leveraged to effectively solve SPGs in practice, some improvements are possible. First, we have not formally examined the optimal arrangement of $\alpha$. It is possible to find the weakest requirements on $\alpha$ so that the reductions are correct, thus optimizing possible implementations. Second, some specific cases lead to very small values of $\alpha$. These cases are similar to the ones that can challenge classical MDP solvers using VI, and so we can benefit from any family of arena structure where these cases are avoided, leading to implementable valuations of $\alpha$.  

\newpage
\bibliography{literature}

\begin{thebibliography}{10}

\bibitem{amir2003stochastic}
Rabah Amir.
\newblock Stochastic games in economics and related fields: an overview.
\newblock {\em Stochastic {G}ames and {A}pplications}, pages 455--470, 2003.

\bibitem{andersson2009complexity}
Daniel Andersson and Peter~Bro Miltersen.
\newblock The complexity of solving stochastic games on graphs.
\newblock In {\em International Symposium on Algorithms and Computation}, pages
  112--121. Springer, 2009.

\bibitem{azeem2022optimistic}
Muqsit Azeem, Alexandros Evangelidis, Jan K{\v{r}}et{\'\i}nsk{\`y}, Alexander
  Slivinskiy, and Maximilian Weininger.
\newblock Optimistic and topological value iteration for simple stochastic
  games.
\newblock In {\em International Symposium on Automated Technology for
  Verification and Analysis}, pages 285--302. Springer, 2022.

\bibitem{baier2008principles}
Christel Baier and Joost-Pieter Katoen.
\newblock {\em Principles of {M}odel {C}hecking}.
\newblock MIT press, 2008.

\bibitem{baier2017ensuring}
Christel Baier, Joachim Klein, Linda Leuschner, David Parker, and Sascha
  Wunderlich.
\newblock Ensuring the reliability of your model checker: Interval iteration
  for {M}arkov decision processes.
\newblock In {\em International Conference on Computer Aided Verification},
  volume 10426, pages 160--180. Springer, 2017.

\bibitem{DBLP:journals/lmcs/BouyerORV23}
Patricia Bouyer, Youssouf Oualhadj, Mickael Randour, and Pierre Vandenhove.
\newblock Arena-independent finite-memory determinacy in stochastic games.
\newblock {\em Log. Methods Comput. Sci.}, 19(4), 2023.
\newblock URL: \url{https://doi.org/10.46298/lmcs-19(4:18)2023}, \href
  {https://doi.org/10.46298/LMCS-19(4:18)2023}
  {\path{doi:10.46298/LMCS-19(4:18)2023}}.

\bibitem{calude2017deciding}
Cristian~S Calude, Sanjay Jain, Bakhadyr Khoussainov, Wei Li, and Frank
  Stephan.
\newblock Deciding parity games in quasipolynomial time.
\newblock In {\em Proceedings of the 49th Annual ACM SIGACT Symposium on Theory
  of Computing}, pages 252--263, 2017.

\bibitem{Chatterjee_2011}
Krishnendu Chatterjee and Nathanaël Fijalkow.
\newblock A reduction from parity games to simple stochastic games.
\newblock {\em Electronic Proceedings in Theoretical Computer Science},
  54:74--86, 2011.
\newblock URL: \url{https://doi.org/10.4204%2Feptcs.54.6}, \href
  {https://doi.org/10.4204/eptcs.54.6} {\path{doi:10.4204/eptcs.54.6}}.

\bibitem{chatterjee2006strategy}
Krishnendu Chatterjee and Thomas~A Henzinger.
\newblock Strategy improvement and randomized subexponential algorithms for
  stochastic parity games.
\newblock In {\em Annual Symposium on Theoretical Aspects of Computer Science},
  pages 512--523. Springer, 2006.

\bibitem{chatterjee2008reduction}
Krishnendu Chatterjee and Thomas~A Henzinger.
\newblock Reduction of stochastic parity to stochastic mean-payoff games.
\newblock {\em Information Processing Letters}, 106(1):1--7, 2008.

\bibitem{chatterjee2010gist}
Krishnendu Chatterjee, Thomas~A Henzinger, Barbara Jobstmann, and Arjun
  Radhakrishna.
\newblock Gist: A solver for probabilistic games.
\newblock In {\em Computer Aided Verification: 22nd International Conference},
  volume 6174, pages 665--669. Springer, 2010.

\bibitem{chatterjee2003sspg}
Krishnendu Chatterjee, Marcin Jurdzi{\'{n}}ski, and Thomas~A. Henzinger.
\newblock Simple stochastic parity games.
\newblock In Matthias Baaz and Johann~A. Makowsky, editors, {\em Computer
  Science Logic}, volume 2803, pages 100--113, Berlin, Heidelberg, 2003.
  Springer Berlin Heidelberg.

\bibitem{chatterjee2004quantitative}
Krishnendu Chatterjee, Marcin Jurdzinski, and Thomas~A Henzinger.
\newblock Quantitative stochastic parity games.
\newblock In {\em SODA}, volume~4, pages 121--130, 2004.

\bibitem{DBLP:conf/csl/ChatterjeeMJ04}
Krishnendu Chatterjee, Rupak Majumdar, and Marcin Jurdzinski.
\newblock On {N}ash equilibria in stochastic games.
\newblock In {\em Computer Science Logic}, volume 3210 of {\em Lecture Notes in
  Computer Science}, pages 26--40. Springer, 2004.
\newblock \href {https://doi.org/10.1007/978-3-540-30124-0\_6}
  {\path{doi:10.1007/978-3-540-30124-0\_6}}.

\bibitem{DBLP:conf/concur/ChatterjeeP19}
Krishnendu Chatterjee and Nir Piterman.
\newblock Combinations of qualitative winning for stochastic parity games.
\newblock In {\em 30th International Conference on Concurrency Theory}, volume
  140 of {\em LIPIcs}, pages 6:1--6:17. Schloss Dagstuhl - Leibniz-Zentrum
  f{\"{u}}r Informatik, 2019.
\newblock URL: \url{https://doi.org/10.4230/LIPIcs.CONCUR.2019.6}, \href
  {https://doi.org/10.4230/LIPICS.CONCUR.2019.6}
  {\path{doi:10.4230/LIPICS.CONCUR.2019.6}}.

\bibitem{condon1990algorithms}
Anne Condon.
\newblock On algorithms for simple stochastic games.
\newblock {\em Advances in {C}omputational {C}omplexity {T}heory}, 13:51--72,
  1990.

\bibitem{condon1992complexity}
Anne Condon.
\newblock The complexity of stochastic games.
\newblock {\em Information and Computation}, 96(2):203--224, 1992.

\bibitem{czerwinski2019universal}
Wojciech Czerwi{\'n}ski, Laure Daviaud, Nathana{\"e}l Fijalkow, Marcin
  Jurdzi{\'n}ski, Ranko Lazi{\'c}, and Pawe{\l} Parys.
\newblock Universal trees grow inside separating automata: Quasi-polynomial
  lower bounds for parity games.
\newblock In {\em Proceedings of the Thirtieth Annual ACM-SIAM Symposium on
  Discrete Algorithms}, pages 2333--2349. SIAM, 2019.

\bibitem{darlington2023stochastic}
Matthew Darlington, Kevin~D. Glazebrook, David~S. Leslie, Rob Shone, and
  Roberto Szechtman.
\newblock A stochastic game framework for patrolling a border.
\newblock {\em Eur. J. Oper. Res.}, 311(3):1146--1158, 2023.
\newblock URL: \url{https://doi.org/10.1016/j.ejor.2023.06.011}, \href
  {https://doi.org/10.1016/J.EJOR.2023.06.011}
  {\path{doi:10.1016/J.EJOR.2023.06.011}}.

\bibitem{eisentraut2022value}
Julia Eisentraut, Edon Kelmendi, Jan K{\v{r}}et{\'\i}nsk{\`y}, and Maximilian
  Weininger.
\newblock Value iteration for simple stochastic games: Stopping criterion and
  learning algorithm.
\newblock {\em Information and Computation}, 285:104886, 2022.

\bibitem{DBLP:journals/lmcs/EtessamiKVY08}
Kousha Etessami, Marta~Z. Kwiatkowska, Moshe~Y. Vardi, and Mihalis Yannakakis.
\newblock Multi-objective model checking of {M}arkov decision processes.
\newblock {\em Log. Methods Comput. Sci.}, 4(4), 2008.
\newblock \href {https://doi.org/10.2168/LMCS-4(4:8)2008}
  {\path{doi:10.2168/LMCS-4(4:8)2008}}.

\bibitem{hahn2011probabilistic}
Ernst~Moritz Hahn, Holger Hermanns, and Lijun Zhang.
\newblock Probabilistic reachability for parametric {M}arkov models.
\newblock {\em International Journal on Software Tools for Technology
  Transfer}, 13:3--19, 2011.

\bibitem{hahn2020model}
Ernst~Moritz Hahn, Mateo Perez, Sven Schewe, Fabio Somenzi, Ashutosh Trivedi,
  and Dominik Wojtczak.
\newblock Model-free reinforcement learning for stochastic parity games.
\newblock In {\em 31st International Conference on Concurrency Theory}, volume
  171 of {\em LIPIcs}, pages 21:1--21:16. Schloss Dagstuhl - Leibniz-Zentrum
  f{\"{u}}r Informatik, 2020.
\newblock URL: \url{https://doi.org/10.4230/LIPIcs.CONCUR.2020.21}, \href
  {https://doi.org/10.4230/LIPICS.CONCUR.2020.21}
  {\path{doi:10.4230/LIPICS.CONCUR.2020.21}}.

\bibitem{hartmanns2020optimistic}
Arnd Hartmanns and Benjamin~Lucien Kaminski.
\newblock Optimistic value iteration.
\newblock In {\em Computer Aided Verification - 32nd International Conference},
  volume 12225 of {\em Lecture Notes in Computer Science}, pages 488--511.
  Springer, 2020.
\newblock \href {https://doi.org/10.1007/978-3-030-53291-8\_26}
  {\path{doi:10.1007/978-3-030-53291-8\_26}}.

\bibitem{hensel2021probabilistic}
Christian Hensel, Sebastian Junges, Joost-Pieter Katoen, Tim Quatmann, and
  Matthias Volk.
\newblock The probabilistic model checker storm.
\newblock {\em International Journal on Software Tools for Technology
  Transfer}, pages 1--22, 2021.

\bibitem{jurdzinski1998deciding}
Marcin Jurdzi{\'n}ski.
\newblock Deciding the winner in parity games is in $\textbf{UP} \cap
  \textbf{co-UP}$.
\newblock {\em Information Processing Letters}, 68(3):119--124, 1998.

\bibitem{jurdzinski2017succinct}
Marcin Jurdzi{\'n}ski and Ranko Lazi{\'c}.
\newblock Succinct progress measures for solving parity games.
\newblock In {\em 2017 32nd Annual ACM/IEEE Symposium on Logic in Computer
  Science (LICS)}, pages 1--9. IEEE, 2017.

\bibitem{kattenbelt2010game}
Mark Kattenbelt, Marta Kwiatkowska, Gethin Norman, and David Parker.
\newblock A game-based abstraction-refinement framework for {M}arkov decision
  processes.
\newblock {\em Formal Methods in System Design}, 36:246--280, 2010.

\bibitem{kvretinsky2022comparison}
Jan K{\v{r}}et{\'\i}nsk{\`y}, Emanuel Ramneantu, Alexander Slivinskiy, and
  Maximilian Weininger.
\newblock Comparison of algorithms for simple stochastic games.
\newblock {\em Information and Computation}, 289:104885, 2022.

\bibitem{kwiatkowska2020prism}
Marta Kwiatkowska, Gethin Norman, David Parker, and Gabriel Santos.
\newblock Prism-games 3.0: Stochastic game verification with concurrency,
  equilibria and time.
\newblock In {\em Computer Aided Verification - 32nd International Conference},
  volume 12225 of {\em Lecture Notes in Computer Science}, pages 475--487.
  Springer, 2020.
\newblock \href {https://doi.org/10.1007/978-3-030-53291-8\_25}
  {\path{doi:10.1007/978-3-030-53291-8\_25}}.

\bibitem{kwiatkowska2011prism}
Marta~Z. Kwiatkowska, Gethin Norman, and David Parker.
\newblock {PRISM} 4.0: Verification of probabilistic real-time systems.
\newblock In {\em Computer Aided Verification - 23rd International Conference},
  volume 6806 of {\em Lecture Notes in Computer Science}, pages 585--591.
  Springer, 2011.
\newblock \href {https://doi.org/10.1007/978-3-642-22110-1\_47}
  {\path{doi:10.1007/978-3-642-22110-1\_47}}.

\bibitem{lehtinen2018modal}
Karoliina Lehtinen.
\newblock A modal $\mu$ perspective on solving parity games in quasi-polynomial
  time.
\newblock In {\em Proceedings of the 33rd Annual ACM/IEEE Symposium on Logic in
  Computer Science}, pages 639--648, 2018.

\bibitem{leslie2020best}
David~S Leslie, Steven Perkins, and Zibo Xu.
\newblock Best-response dynamics in zero-sum stochastic games.
\newblock {\em Journal of Economic Theory}, 189:105095, 2020.

\bibitem{DBLP:conf/adhs/MajumdarMSS21}
Rupak Majumdar, Kaushik Mallik, Anne{-}Kathrin Schmuck, and Sadegh Soudjani.
\newblock Symbolic qualitative control for stochastic systems via finite parity
  games.
\newblock In {\em 7th {IFAC} Conference on Analysis and Design of Hybrid
  Systems}, volume~54 of {\em IFAC-PapersOnLine}, pages 127--132. Elsevier,
  2021.
\newblock URL: \url{https://doi.org/10.1016/j.ifacol.2021.08.486}, \href
  {https://doi.org/10.1016/J.IFACOL.2021.08.486}
  {\path{doi:10.1016/J.IFACOL.2021.08.486}}.

\bibitem{martin1998determinacy}
Donald~A Martin.
\newblock The determinacy of {B}lackwell games.
\newblock {\em The Journal of Symbolic Logic}, 63(4):1565--1581, 1998.

\bibitem{parys2019parity}
Pawel Parys.
\newblock Parity games: Zielonka's algorithm in quasi-polynomial time.
\newblock In {\em 44th International Symposium on Mathematical Foundations of
  Computer Science}, volume 138 of {\em LIPIcs}, pages 10:1--10:13. Schloss
  Dagstuhl - Leibniz-Zentrum f{\"{u}}r Informatik, 2019.
\newblock URL: \url{https://doi.org/10.4230/LIPIcs.MFCS.2019.10}, \href
  {https://doi.org/10.4230/LIPICS.MFCS.2019.10}
  {\path{doi:10.4230/LIPICS.MFCS.2019.10}}.

\bibitem{phalakarn2020widest}
Kittiphon Phalakarn, Toru Takisaka, Thomas Haas, and Ichiro Hasuo.
\newblock Widest paths and global propagation in bounded value iteration for
  stochastic games.
\newblock In {\em Computer Aided Verification: 32nd International Conference},
  pages 349--371. Springer, 2020.

\bibitem{puterman2014markov}
Martin~L Puterman.
\newblock {\em {M}arkov decision processes: discrete stochastic dynamic
  programming}.
\newblock John Wiley \& Sons, 2014.

\bibitem{quatmann2018sound}
Tim Quatmann and Joost{-}Pieter Katoen.
\newblock Sound value iteration.
\newblock In {\em Computer Aided Verification - 30th International Conference},
  volume 10981 of {\em Lecture Notes in Computer Science}, pages 643--661.
  Springer, 2018.
\newblock \href {https://doi.org/10.1007/978-3-319-96145-3\_37}
  {\path{doi:10.1007/978-3-319-96145-3\_37}}.

\bibitem{shapley1953stochastic}
Lloyd~S Shapley.
\newblock Stochastic games.
\newblock {\em Proceedings of the {N}ational {A}cademy of {S}ciences},
  39(10):1095--1100, 1953.

\bibitem{DBLP:journals/tcs/SimardDL21}
Fr{\'{e}}d{\'{e}}ric Simard, Jos{\'{e}}e Desharnais, and Fran{\c{c}}ois
  Laviolette.
\newblock General cops and robbers games with randomness.
\newblock {\em Theor. Comput. Sci.}, 887:30--50, 2021.
\newblock URL: \url{https://doi.org/10.1016/j.tcs.2021.06.043}, \href
  {https://doi.org/10.1016/J.TCS.2021.06.043}
  {\path{doi:10.1016/J.TCS.2021.06.043}}.

\bibitem{DBLP:journals/corr/abs-1109-4017}
Michael Ummels and Dominik Wojtczak.
\newblock The complexity of {N}ash equilibria in stochastic multiplayer games.
\newblock {\em Log. Methods Comput. Sci.}, 7(3), 2011.
\newblock \href {https://doi.org/10.2168/LMCS-7(3:20)2011}
  {\path{doi:10.2168/LMCS-7(3:20)2011}}.

\bibitem{zwick1996complexity}
Uri Zwick and Mike Paterson.
\newblock The complexity of mean payoff games on graphs.
\newblock {\em Theoretical Computer Science}, 158(1-2):343--359, 1996.

\end{thebibliography}
\newpage
\appendix
\section{Appendix}

\subsection{A General Lemma}
\label{app-subsec:lower_bound_general}
\begin{lemma}[Lower Bound of Sink Reachability]\label{lm:lower_bound_general} 
Let $\mc=(V\uplus\{v_f, v_s, v_b\}, \mcprob, v_I)$ be a Markov Chain with $m$ states.
If the following conditions hold:
\begin{enumerate}
    \item There are only two BSCCs, namely $\{v_s\}$ and $\{v_b\}$, and $\mcprob(v_s,v_s)= \mcprob(v_b,v_b) =1$.
    \item\label{cond:tri-entries} For all states $v \in V$, 
    \begin{itemize}
        \item $\mcprob(v, v_s) = \alpha$, $\mcprob(v, V\uplus\{v_f\}) = 1-\alpha$, $\mcprob(v, v_b) = 0$. 
        \item $\Pr^{v}(\Reach(v_b))>0$.
        \item For all $v' \in V\uplus \{v_f\}$, $\mcprob(v,v')>0$ implies that $\mcprob(v,v')\ge s$.
    \end{itemize}  
    \item For $v_f$, $\mcprob(v_f,V\uplus\{v_f\}) = k$ and $\mcprob(v_f, v_b) = l$.
\end{enumerate}
Then for all $v \in V$, we can obtain a lower bound for the reachability probability of $v_b$ as follows:
\[
\Pr^{v}(\Reach(v_b)) \ge \frac{(1-s)\cdot s^{m}\cdot l}{1-(s+t) + k\cdot s^{m+1} + t\cdot s^{m} - k \cdot s^{m}}
\]
where $t=1-\alpha_0-s$.
\end{lemma}

\begin{proof} We assume all conditions in Lemma~\ref{lm:lower_bound_general} hold.

For all $v \in V$, a play starting from $v$ has to reach $v_f$ first to reach $v_b$, and therefore:
\[
\Pr^{v}(\Reach(v_b)) = \Pr^{v}(\Reach(v_f)) \cdot \Pr^{v_f}(\Reach(v_b)) \le \Pr^{v_f}(\Reach(v_b))
\]
We rename all $v \in V$ in the order of the probability for a play starting from $v$ to reach $v_b$, such that:
\begin{equation}\label{eq:worstcase-before}
0 < \Pr^{v_1}(\Reach(v_b)) \le \dots \le \Pr^{v_m}(\Reach(v_b)) \le \Pr^{v_f}(\Reach(v_b))    
\end{equation}
For all $i=1,2,\cdots, m$, we denote $\Pr^{v_i}(\Reach(v_b))$ with $p_i$, and we denote $\Pr^{v_f}(\Reach(v_b))$ with $p_{m+1}$, so we write inequality~\ref{eq:worstcase-before} as: 
\[
0<p_1\le p_2 \le \dots \le p_m \le p_{m+1}
\]
It follows from condition~\ref{cond:tri-entries} that for all $i=1,2,\cdots, m$:
\begin{align*}
    & \mcprob(v_i,v_s) = \alpha,\ \mcprob(v_i, v_b) = 0  \\
    & \sum^{m+1}_{j=1} \mcprob(v_i,v_j) = 1-\alpha
\end{align*}
We recall Theorem~\ref{th:reachability} for the linear equation system for calculating reachability probabilities in MCs. As there is no direct transition from $v_i$ to $v_b$, we write $p_i$ as:
\begin{equation}\label{eq:mc_worst1}
p_i = \sum^{m+1}_{j=1}(\mcprob(v_i,v_j)\cdot p_j)    
\end{equation}
    
It follows from condition~\ref{cond:tri-entries} that $\sum^{m+1}_{j=1}\mcprob(v_i,v_j) = 1-\alpha$ and $p_i>0$. We claim that for equation~\ref{eq:mc_worst1} to hold, there must exist $j>i$ such that $\delta(v_i,v_j)>0$. 
Therefore we obtain the following for $p_i$:
\begin{align}\label{eq:chain-pre}
    p_i &=\ \sum^{m+1}_{j=1}(\mcprob(v_i,v_j)\cdot p_j) \notag\\
    & =\ \sum^i_{j=1}(\mcprob(v_i,v_j)\cdot p_j) + \sum^{m+1}_{j=i+1}(\mcprob(v_i,v_j)\cdot p_j) \notag \\
    & \ge\ (\sum^i_{j=1}\mcprob(v_i,v_j))\cdot p_1 + (\sum^{m+1}_{j=i+1}\mcprob(v_i,v_j))\cdot p_{i+1} \notag \\
    & \ge\ t p_1 + s p_{i+1} 
\end{align}
where $t = 1-\alpha-s$.

We rearrange inequality~\ref{eq:chain-pre} into the following form:
\begin{equation}\label{eq:chain-first}
p_i - \frac{t}{1-s}p_1 \ge s(p_{i+1} - \frac{t}{1-s}p_1)    
\end{equation}

Similarly, for $p_{m+1}$ we have the following:
\begin{align}\label{eq:chain-second}
    p_{m+1} &=\ \sum^{m+1}_{j=1}(\mcprob(v_f,v_j)\cdot p_j) + \mcprob(v_f,v_b) \notag\\ 
    & \ge\ (\sum^{m+1}_{j=1}\mcprob(v_f,v_j))\cdot p_1 +  \mcprob(v_f,v_b) \notag\\ 
    & =\ k p_1 + l 
\end{align}
We combine inequalities~\ref{eq:chain-first} and~\ref{eq:chain-second} together to obtain the following:
\begin{align}\label{eq:chain-third}
    p_1 - \frac{t}{1-s}p_1 \notag
    & \ge\ s(p_{2} - \frac{t}{1-s}p_1) \notag\\
    & \ge\ s^2(p_{3} - \frac{t}{1-s}p_1) \notag\\
    & \ge\ \dots \notag\\
    & \ge\ s^{m} (p_{m+1} - \frac{t}{1-s}p_1) \notag\\
    & \ge\ s^{m} (k p_1 + l - \frac{t}{1-s}p_1) 
\end{align}
We rearrange~\eqref{eq:chain-third} to obtain:
\[
p_1 \ge  \frac{(1-s)\cdot s^{m}\cdot l}{1-(s+t) + k\cdot s^{m+1} + t\cdot s^{m} - k \cdot s^{m}}
\]
\end{proof}
\begin{remark}[Actual Worst Case]
In the classical example of Figure~\ref{fig:worst-case-mc}, where $m=5$, the inequality of Lemma~\ref{lm:lower_bound_general} becomes an equality. It can easily be generalized to an arbitrary number of states to obtain the general worst-case. We draw the transitions to $v_s$ with dotted arrows for visual neatness. Note that Lemma~\ref{lm:lower_bound_general} applies no matter $\delta(v_f,v_s)>0$ or not.
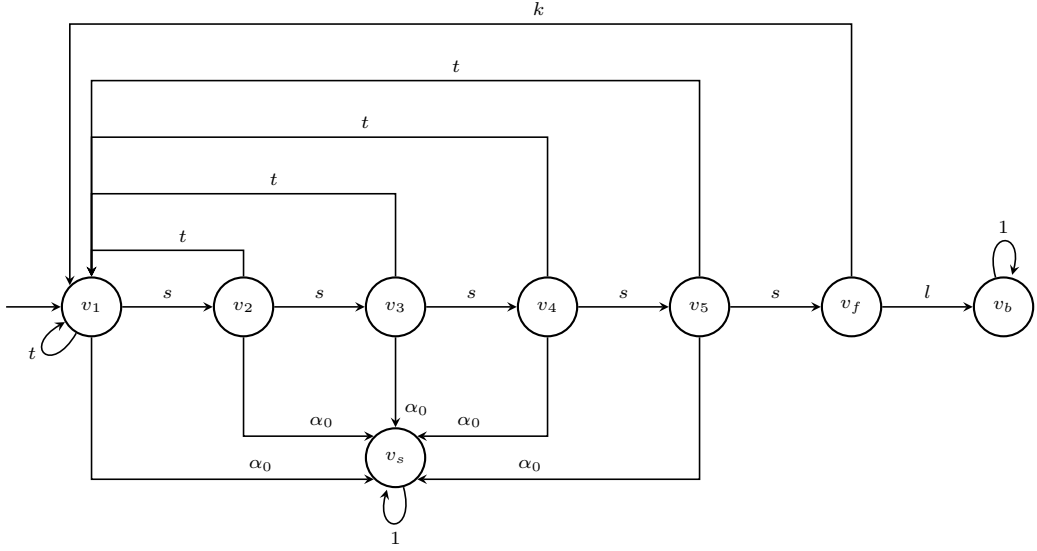
\begin{figure}
    \centering
    \begin{tikzpicture}[]
    \node[circle] (-1) at (-1.3,0) {};
    \node[plainNodes] (0) at (0,0) {$v_1$};
    \node[plainNodes] (1) at (2,0) {$v_2$};
    \node[plainNodes] (2) at (4,0) {$v_3$};
    \node[plainNodes] (3) at (6,0) {$v_4$};
    \node[plainNodes] (4) at (8,0) {$v_5$};
    \node[plainNodes] (5) at (10,0) {$v_f$};
    \node[plainNodes] (6) at (12,0) {$v_b$};
    \node[plainNodes] (7) at (4,-2) {$v_s$};
    
    \draw[->] (-1) -- (0);
    \draw[->] (0) -- node[labelNodes,above] {$s$} (1);
    \draw[->] (1) -- node[labelNodes,above] {$s$} (2);
    \draw[->] (2) -- node[labelNodes,above] {$s$} (3);
    \draw[->] (3) -- node[labelNodes,above] {$s$} (4);
    \draw[->] (4) -- node[labelNodes,above] {$s$} (5);
    \draw[->] (5) -- node[labelNodes,above] {$l$} (6);

    \draw[->] (0) |- node[above, pos=0.8, labelNodes] {$\alpha_0$} (7.south west);
    \draw[->] (1) |- node[above, pos=0.8, labelNodes] {$\alpha_0$} (7.north west);
    \draw[->] (2) -- node[right, pos=0.8, labelNodes] {$\alpha_0$} (7);
    \draw[->] (3) |- node[above, pos=0.8, labelNodes] {$\alpha_0$} (7.north east);
    \draw[->] (4) |- node[above, pos=0.8, labelNodes] {$\alpha_0$} (7.south east);

    \draw[->] (0) edge[loop, in=210, out=240, looseness=8] node[labelNodes,left] {$t$} (0);
    \draw[->] (1) -- ++(0,0.75) -| node[labelNodes, pos=0.2, above] {$t$} (0);
    \draw[->] (2) -- ++(0,1.5) -| node[labelNodes, pos=0.2, above] {$t$} (0);
    \draw[->] (3) -- ++(0,2.25) -| node[labelNodes, pos=0.2, above] {$t$} (0);
    \draw[->] (4) -- ++(0,3) -| node[labelNodes, pos=0.2, above] {$t$} (0);
    \draw[->] (5) -- ++(0,3.75) -| node[labelNodes, pos=0.2, above] {$k$} (0.north west);
    \draw (6) edge[loop above] node[labelNodes] {$1$} (6);
    \draw (7) edge[loop below] node[labelNodes] {$1$} (7);
    \end{tikzpicture}
    \caption{Worst-case Markov chain}
    \label{fig:worst-case-mc}
\end{figure}
\end{remark}

\subsection{Proof of Lemma~\ref{lm:formal-crosspath}}
\label{app-subsec:crosspath}
\lowercross* 
To prove Lemma~\ref{lm:formal-crosspath}, we first fix an arbitrary strategy pair $\sigma,\gamma\in\Sigma\times\Gamma$ and consider the Markov chains $\mc_{\sigma,\gamma}$ and $\omc_{\sigma,\gamma}$ associated to $G_{\sigma,\gamma}$ and $\wG_{\sigma,\gamma}$. 
Note that in the following part of this proof, we leave out the initial state of the Markov chains. Hence we have $\mc_{\sigma,\gamma}=(V,\mcprob)$ and $\omc_{\sigma,\gamma} = (\oV\uplus \widehat{V}\uplus\{\vwin,\vlose\}, \omcprob)$. Additionally, we denote the set of vertices that belongs to BSCCs in $\mc_{\sigma,\gamma}$ with $V_{B}$ and let $V_{T} = V \backslash V_{B}$. Correspondingly, the set of vertices that belong to pBSCCs in $\omc_{\sigma,\gamma}$ is $\oV_{B}\uplus\widehat{V}_{B}$.

\begin{figure}

\begin{subfigure}[b]{\textwidth}
\centering
\begin{tikzpicture}[scale=1.1]
\node[circle,draw=none] (s) at (2,1) {};
\node[plainNodes] (0) at (2,0) {$\ov_0$};
\draw[->] (s) -- (0);
\node[hatNodes] (00) at (0,0) {$\widehat{v}_0$};
\node[plainNodes] (2) at (6,-4.5) {$\ov_1$};
\node[hatNodes] (02) at (4,-4.5) {$\widehat{v}_1$};
\node[plainNodes] (3) at (6,0) {$\ov_2$};
\node[hatNodes] (03) at (4,0) {$\widehat{v}_2$};
\node[plainNodes] (4) at (6,-1.5) {$\ov_3$};
\node[hatNodes] (04) at (4,-1.5) {$\widehat{v}_3$};
\node[plainNodes] (5) at (10,0) {$\ov_4$};
\node[hatNodes] (05) at (8,0) {$\widehat{v}_4$};

\node[winNode] (win) at (4,1.5) {$\vwin$};
\node[loseNode] (lose) at (5.5,-3) {$\vlose$};

\draw[->] (0) -- node[labelNodes,right,pos=0.3] {$\frac{1}{2}$} (02);
\draw[->] (0) -- node[labelNodes,above] {$\frac{1}{4}$} (03);
\draw[->] (0) -- node[labelNodes,above,pos=0.6] {$\frac{1}{4}$} (04);
\draw[->] (2) -- node[labelNodes,above,pos=0.8] {$\frac{1}{2}$} (00);
\draw[->] (2) -- node[labelNodes,right] {$\frac{1}{2}$} (05);
\draw[->] (3) -- node[labelNodes,above] {$1$} (05);
\draw[->] (4) -- node[labelNodes,above] {$1$} (03);
\draw[->] (5) to[bend right=45] node[labelNodes,above] {$1$} (05);

\draw[->] (00) -- node[labelNodes, above] {$1{-}\alpha_0$} (0);
\draw[->] (02) -- node[labelNodes, above] {$1{-}\alpha_1$} (2);
\draw[->] (03) -- node[labelNodes, above] {$1{-}\alpha_2$} (3);
\draw[->] (04) -- node[labelNodes, above] {$1{-}\alpha_3$} (4);
\draw[->] (05) -- node[labelNodes, above] {$1{-}\alpha_4$} (5);

\draw[->] (00) to[bend left=15] node[labelNodes, above, pos=0.7] {$\alpha_0$} (win);
\draw[->] (02.north east) -- node[labelNodes, left, pos=0.3] {$\alpha_1$} (lose);
\draw[->] (03) -- node[labelNodes,right, pos=0.7] {$\alpha_2$} (win);
\draw[->] (04) -- node[labelNodes, right, pos=0.3] {$\alpha_3$} (lose);
\draw[->] (05) to[bend right=15] node[labelNodes, above, pos=0.7] {$\alpha_4$} (win);

\draw[->] (win) edge [loop above] node[labelNodes] {$1$} ();
\draw[->] (lose) edge [loop right] node[labelNodes] {$1$} ();
\end{tikzpicture}
\label{subfig:new-sub-arena}
\end{subfigure}
\caption{A sub-arena $\wG_{\sigma,\gamma}$}
\label{fig:step0}
\end{figure}
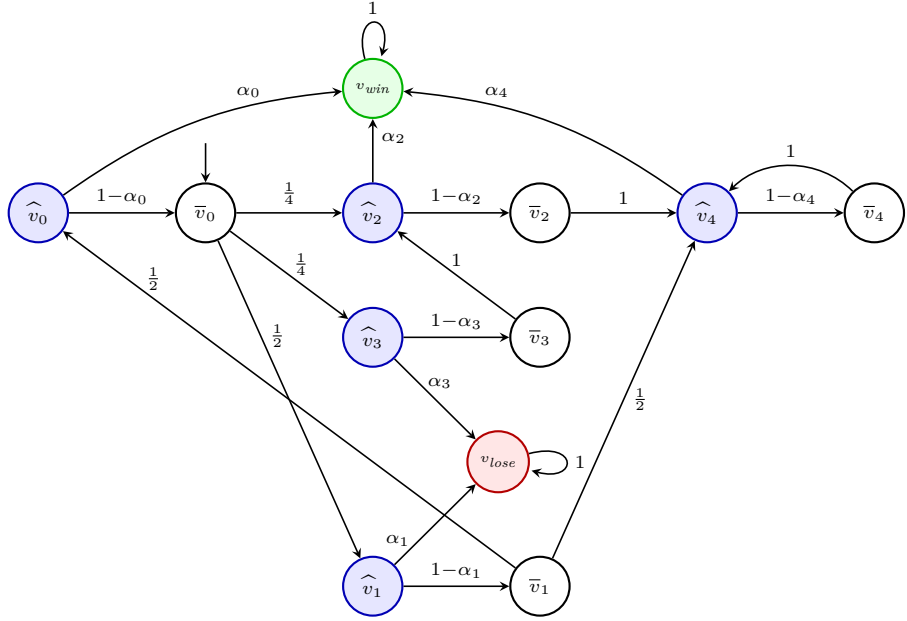
We use the Markov chain in Figure~\ref{fig:step0} as an example. For all vertices $v_i$ in the figure, we let $p(v_i)=i$. We leave out $\vwin$ , $\vlose$  and their incoming transitions, and represent newly introduced vertices with dotted circles for visual neatness. Let us observe that $V_B=\{v_5\}$ and $V_T=\{v_0,v_1,v_2,v_3,v_4\}$.

The proof of Lemma~\ref{lm:formal-crosspath} consists of two parts:
\begin{enumerate}
    \item We make a $4$-step transformation to $\omc_{\sigma,\gamma}$. We denote the resulting Markov Chain of the $i$-th step with $\mc_i = (V_i, \mcprob_i)$. In the $i$-th Markov Chain, we denote with $\Pr_i^v(\Reach(V'_i))$ the probability of reaching $V'_i\subseteq V_i$ from $v$. We show that for all $\ov \in \oV_T$, we have:
    \[\wPr^{\ov}_{\sigma,\gamma}(\crossPath) \ge \Pr^{v}_4(\Reach(v_b))\]
    where $v_b \in V_4$ is a specific vertex in $\mc_4$.
    \item We apply Lemma~\ref{lm:lower_bound_general} to $\mc_4$ to obtain a lower bound of $\Pr^{v}_4(\Reach(v_b))$, and hence a lower bound of $\wPr^{\ov}_{\sigma,\gamma}(\crossPath)$. 
\end{enumerate}
We now present the proof as follows.
\subsubsection{First Part: Transforming the Markov Chain}\label{app-subsec:transform} 

We use the Markov Chains from Figure~\ref{fig:step0} as a running example.
\begin{enumerate}
\item\label{step1} In the first step, we eliminate $\widehat{V}$ according to Lemma 1 in~\cite{hahn2011probabilistic}.  Formally we have $V_1 = \oV\uplus\{\vwin,\vlose\}$, and the resulting transition function $\mcprob_1$ can be described as follows:
\begin{itemize}
    \item For all $\ou,\ov \in \oV\times\oV$, $\mcprob_1(\ou, \ov) = (1-\alpha_{p(v)})\cdot\omcprob(\ou, \widehat{v})$.
    \item For all $\ou \in \oV$, 
        \begin{align*}
            \mcprob_1(\ou, \vwin) = \sum_{p(v)\text{ is even}} \alpha_{p(v)}\cdot\omcprob(\ou,\widehat{v}) \\
            \mcprob_1(\ou, \vlose) = \sum_{p(v)\text{ is odd}} \alpha_{p(v)}\cdot\omcprob(\ou,\widehat{v}).
        \end{align*}
    \item $\mcprob_1(\vwin,\vwin)=\mcprob_1(\vlose,\vlose)=1$.
\end{itemize}

Note that this transformation does not change the probability for a play starting from $\ov\in\oV$ to reach $\vwin$ or $\vlose$. We refer to~\cite{hahn2011probabilistic} for the details of this transformation. 
In the rest of this proof, we rename all vertices $\ov\in\oV$  as $v$ for visual neatness. It follows that $\mc_1 = (V\uplus\{\vwin,\vlose\}, \mcprob_1)$, and for all vertices $v \in V_{T}$:
\[
\wPr^{\ov}_{\sigma,\gamma}(\crossPath) = \Pr_1^{v}(\Reach(V_{B}))
\]
In Figure~\ref{fig:step1}, we present the resulting $\mc_1$ for the running example. For visual neatness, we only draw the transitions of $v_0$, which are the most complicated ones, as a demonstration. 
\begin{figure}[H]
\centering
\begin{tikzpicture}[scale=1.5,
]
\node[circle,draw=none] (s) at (-1.2,0) {};
\node[plainNodes] (0) at (0,0) {$v_0$};
\draw[->] (s) -- (0);
\node[plainNodes] (2) at (4,-2.5) {$v_1$};
\node[plainNodes] (3) at (4,0) {$v_2$};
\node[plainNodes] (4) at (4,-1.5) {$v_3$};
\node[winNode] (win) at (0,1.5) {$\vwin$};
\node[loseNode] (lose) at (0,-1.5) {$\vlose$};

\draw[->] (0) -- node[labelNodes,sloped,above] {$\frac{1}{2}(1{-}\alpha_1)$} (2);
\draw[->] (0) -- node[labelNodes,left] {$\frac{1}{4}\alpha_2$} (win);
\draw[->] (0) -- node[labelNodes,left] {$\frac{1}{2}\alpha_1{+}\frac{1}{4}\alpha_3$} (lose);

\draw[->] (0) -- node[labelNodes,sloped,above] {$\frac{1}{4}(1{-}\alpha_2)$} (3);
\draw[->] (0) -- node[labelNodes,sloped,above] {$\frac{1}{4}(1{-}\alpha_3)$} (4);
\end{tikzpicture}
\caption{Before entering an BSCC: eliminating $\widehat{V}$}
\label{fig:step1}
\end{figure}
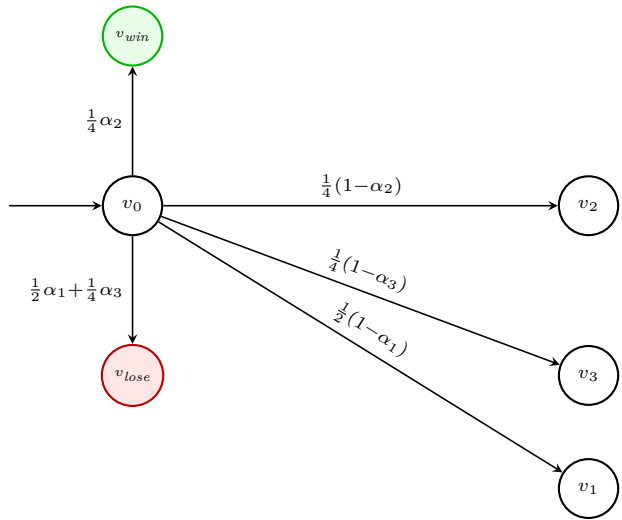
\begin{observation}\label{obs:step1}
Let us recall the original Markov chain $\mc_{\sigma,\gamma} = (V,\mcprob)$. An Effect of this transformation is that for all pairs of vertices $u,v\in V\times V$, $\mcprob_1(u,v)=(1-\alpha_{p(v)})\mcprob(u,v)$. Intuitively, if we do not consider $\vwin$ and $\vlose$, $\mc_1$ has the same transitions as $\mc_{\sigma,\gamma}$ with discounted probabilities, and it follows that for all $v\in V_T$, $\Pr_1^{v}(\Reach(V_{B}))>0$.
\end{observation}

\item\label{step2} In this step, we scale up the values of $\alpha$ to its maximum. Specifically, for all $n\in \mathbb{N}$, we let $\alpha_{n} = \alpha_{0}$. The vertices of the Markov Chain remain the same, hence $V_2 = V\uplus\{\vwin,\vlose\}$. We obtain the resulting transition function $\mcprob_2$ by substituting every occurrence of $\alpha_{p(v)}$ in $\mcprob_1$ with $\alpha_{0}$. Since we scale up the values of $\alpha$, it is observable that for all $v\in V_T$, the probability of reaching $\vwin$ or $\vlose$ from $v$ does not decrease, and thus:
\[
\Pr_1^{v}(\Reach(V_{B})) \ge \Pr_2^{v}(\Reach(V_{B}))
\]
In Figure~\ref{fig:step2}, we present $v_0$ and its outgoing transitions in the resulting $\mc_2$ for the running example. The same idea applies to all other vertices.
\begin{figure}[H]
\centering
\begin{tikzpicture}[scale=1.5,
]
\node[circle,draw=none] (s) at (-1.2,0) {};
\node[plainNodes] (0) at (0,0) {$v_0$};
\draw[->] (s) -- (0);
\node[plainNodes] (2) at (4,-2.5) {$v_1$};
\node[plainNodes] (3) at (4,0) {$v_2$};
\node[plainNodes] (4) at (4,-1.5) {$v_3$};
\node[winNode] (win) at (0,1.5) {$\vwin$};
\node[loseNode] (lose) at (0,-1.5) {$\vlose$};

\draw[->] (0) -- node[labelNodes,sloped,above] {$\frac{1}{2}(1{-}\alpha_0)$} (2);
\draw[->] (0) -- node[labelNodes,left] {$\frac{1}{4}\alpha_0$} (win);
\draw[->] (0) -- node[labelNodes,left] {$\frac{1}{2}\alpha_1{+}\frac{1}{4}\alpha_0$} (lose);

\draw[->] (0) -- node[labelNodes,sloped,above] {$\frac{1}{4}(1{-}\alpha_0)$} (3);
\draw[->] (0) -- node[labelNodes,sloped,above] {$\frac{1}{4}(1{-}\alpha_0)$} (4);
\end{tikzpicture}
\caption{Before entering an BSCC: scaling up $\alpha$}
\label{fig:step2}
\end{figure}
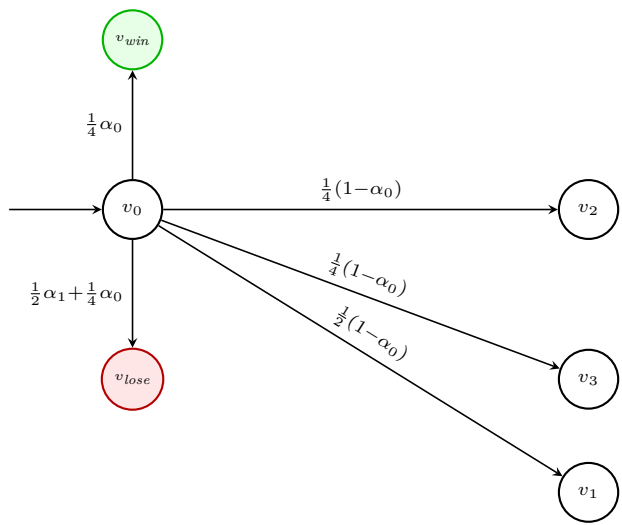
\begin{observation}\label{obs:step2}
\begin{itemize}
    \item For all vertices $v\in V$, we have $\mcprob_2(v,\{\vwin,\vlose\}) = \alpha_{0}$ and $\mcprob_2(v,V)=1-\alpha_{0}$.
    \item For all pairs of vertices $u,v\in V\times V$, $\mcprob_2(u,v) = (1-\alpha_{0})\mcprob(u,v)$. If we do not consider $\vwin$ and $\vlose$, $\mc_2$ still has the same transitions as $\mc_{\sigma,\gamma}$ with discounted probabilities, and it follows that for all $v\in V_T$, $\Pr_2^{v}(\Reach(V_{B}))>0$.
\end{itemize}    
\end{observation}

\item\label{step3} In this step, we merge $\vwin$ and $\vlose$ as a single sink $v_s$ with a self-loop of probability $1$. The subscript $s$ here stands for \textit{sinks}. All incoming transitions to either $\vwin$ or $\vlose$ now go to $v_s$. 

Additionally, we merge all vertices in $V_B$ as a single vertex. To give a clear intuition, we consider this as two sub-steps:
\begin{itemize}
    \item We first duplicate $v_s$, and denote the new one with $v'_s$. For all vertices $v\in V_{B}$, we replace the transitions $(v,v_s)$ with transitions $(v,v'_s)$.
    \item It follows from the construction that there is no transition leaving $V_{B}\uplus\{v'_s\}$. We collapse $V_{B}\uplus\{v'_s\}$ into one single vertex $v_b$ with a self-loop of probability $1$. The subscript $b$ here stands for \textit{bottom}. For all vertices $u\in V_{T}$, all transitions from $u$ into $V_{B}$ are merged into transition $(u,v_b)$. 
\end{itemize}
Formally the vertices now become $V_3 = V_{T}\uplus\{v_b,v_s\}$, and the new transition function $\mcprob_3$ can be described as follows:
\begin{itemize}
    \item For all $u, v \in V_T\times V_T$, we have $\mcprob_3(u,v) = \mcprob_2(u,v)$.
    \item For all $u\in V_{T}$, $\mcprob_3(u,v_s) = \mcprob_2(u,\vwin) + \mcprob_2(u,\vlose)=\alpha_{0}$.
    \item For all $u\in V_{T}$, $\mcprob_3(u,v_b) = \sum_{v\in V_{B}} \mcprob_2(u,v)$.
    \item $\mcprob_3(v_s,v_s) = \mcprob_3(v_b, v_b) = 1$.
\end{itemize}
It follows from the transformation that for all vertices $v \in V_{T}$:
\[
\Pr_2^{v}(\Reach(V_{B})) = \Pr_3^{v}(\Reach(v_{b}))
\]
In Figure~\ref{fig:step3}, we present the resulting $\mc_3$ for the running example. 
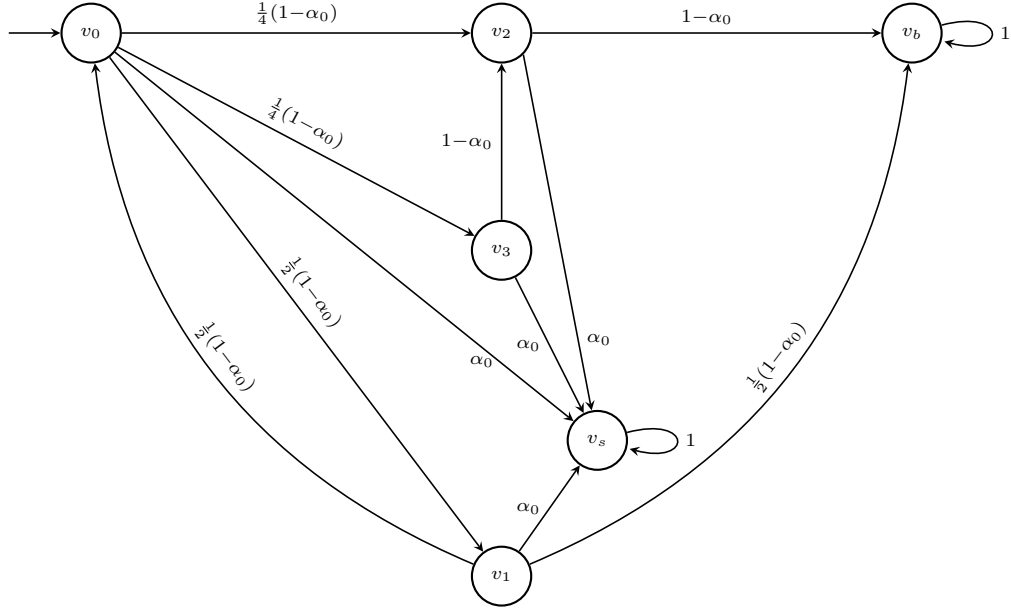
\begin{figure}
\centering
\begin{tikzpicture}[scale=1.8,
]
\node[plainNodes] (0) at (0,4) {$v_0$};
\node[circle, draw=none] (s) at (-0.7, 4) {};
\draw[->] (s) -- (0);
\node[plainNodes] (2) at (3,0) {$v_1$};
\node[plainNodes] (3) at (3,4) {$v_2$};
\node[plainNodes] (4) at (3,2.4) {$v_3$};
\node[plainNodes] (5) at (6,4) {$v_b$};
\node[plainNodes] (6) at (3.7,1) {$v_s$};

\draw[->] (0) -- node[labelNodes,sloped,above] {$\frac{1}{2}(1{-}\alpha_0)$} (2);
\draw[->] (0) -- node[labelNodes,sloped,above] {$\frac{1}{4}(1{-}\alpha_0)$} (3);
\draw[->] (0) -- node[labelNodes,sloped,above] {$\frac{1}{4}(1{-}\alpha_0)$} (4);
\draw[->] (2) to[bend left] node[labelNodes,sloped,above] {$\frac{1}{2}(1{-}\alpha_0)$} (0);
\draw[->] (2) to[bend right] node[labelNodes,sloped,above] {$\frac{1}{2}(1{-}\alpha_0)$} (5);
\draw[->] (3) -- node[labelNodes,sloped,above] {$1{-}\alpha_0$} (5);
\draw[->] (4) -- node[labelNodes,left] {$1{-}\alpha_0$} (3);
\draw[->] (5) edge[loop right] node[labelNodes,] {$1$} (5);

\draw[->] (6) edge [loop right] node[labelNodes,] {$1$} (6);

\draw[->] (0) -- node[labelNodes, below, pos=0.8]{$\alpha_0$} (6);
\draw[->] (2) -- node[labelNodes, left, pos=0.5]{$\alpha_0$} (6);
\draw[->] (3.south east) -- node[labelNodes, right, pos=0.8]{$\alpha_0$} (6);
\draw[->] (4) -- node[labelNodes, left, pos=0.5]{$\alpha_0$} (6);
\end{tikzpicture}
\caption{Before entering an BSCC: merging sinks and BSCCs}
\label{fig:step3}
\end{figure}
\begin{observation}\label{obs:step3}
If we do not consider $\vwin$ and $\vlose$, the non-BSCC part of $\mc_3$ still has the same transitions as the non-BSCC part of $\mc_{\sigma,\gamma}$ with discounted probabilities, and it follows that for all $v\in V_T$, $\Pr_3^{v}(\Reach(v_b))>0$.
\end{observation}

\item Before going into the transformation, we introduce the notion of \textit{frontier vertex}. In the resulting Markov Chain $\mc_3$ of previous steps, we call $v_f \in V_{T}$ a \textit{frontier vertex} if $\mcprob_3(v_f,v_b)>0$, and we denote the set of frontier vertices with $V_{F}$. In the running example, we note that $V_F=\{v_2,v_3\}$.

Since all paths of a finite-state Markov Chain reach a BSCC, for all vertices $v \in V_{T}$, there is at least one path from $v$ to $v_b$. We fix a starting vertex $v_0 \in V_T$. Without loss of generality, we assume all other vertices are reachable from $v_0$. Otherwise, we can always eliminate the unreachable fragment. In the running example, we fix a starting vertex $v_0$, and $v_4$ can be ignored as it is not reachable from $v_0$.

In this step, we apply the following procedure to $\mc_3 = (V_{T}\uplus\{v_b,v_s\}, \mcprob_3)$ to \textit{`defrontierize'} all but one frontier vertices:
\begin{algorithm}[H]
\begin{algorithmic}[1]
\Procedure{DFV}{$\mc = (V_{T}\uplus\{v_b,v_s\}, \mcprob)$}\Comment{ $\mc$ results from~\ref{step3}}
\While{$|V_F| \ge 2$}\Comment{ $\mc$ has more than one frontier vertices}
\State Take an arbitrary vertex $ v_f \in V_F$.
\State Take a vertex $v_p \in V_{T} \backslash V_{F}$, s.t. there exists $\pi = v_p \cdots v'_{f} v_b$ where $v'_f \neq v_f$.
\State $\mcprob(v_f, v_p) \gets \mcprob(v_f, v_p)+\mcprob(v_f,v_b)$
\State $\mcprob(v_f,v_b) \gets 0$
\EndWhile
\State \textbf{return} $\mc$
\EndProcedure
\end{algorithmic}
\end{algorithm}
\begin{remark}[Defrontierize Frontier Vertices]\label{rm:all_can_reach} We note the following regarding the procedure:
\begin{itemize}
    \item Regarding line $4$, there must exist a vertex $v_p \in V_{T} \backslash V_{F}$, such that there exists a path $\pi = v_p \cdots v'_{f} v_b$, where $v'_{f} \neq v_{f}$. Otherwise, it indicates that $v_f$ is the only frontier vertex. The subscript $p$ here stands for \textit{pivot}. Moreover, this ensures that if there is a path $v\cdots v_f v_b$ before an iteration, there is a path $v\cdots v_f v_p \cdots v_b$ after the iteration.
    \item Regarding lines $5-6$, we consider them as the following equivalent version for clarity: 
    \begin{align*}
        & \mcprob' \gets \mcprob \\
        & \mcprob'(v_f, v_p) \gets \mcprob(v_f, v_p)+\mcprob(v_f,v_b) \\
        & \mcprob'(v_f,v_b) \gets 0 \\
        & \mcprob \gets \mcprob'
    \end{align*}
    It follows that $v_f$ is no longer a frontier vertex at the end of the iteration.
\end{itemize}
\end{remark}

In Figure~\ref{fig:step4}, we present the resulting Markov Chain for the running example. 
We defrontierize $v_3$ by substituting transition $(v_3,v_b)$ with transition $(v_3,v_0)$. We draw deleted and added transitions in orange and purple respectively.
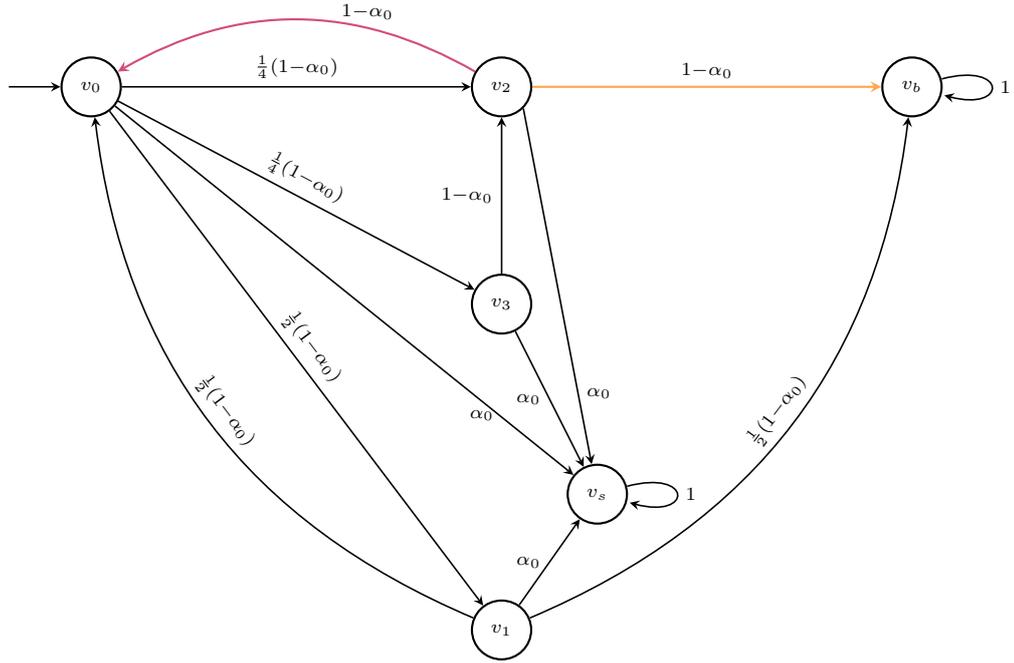
\begin{figure}
\centering
\begin{tikzpicture}[scale=1.8,
]
\node[plainNodes] (0) at (0,4) {$v_0$};
\node[circle, draw=none] (s) at (-0.7, 4) {};
\draw[->] (s) -- (0);
\node[plainNodes] (2) at (3,0) {$v_1$};
\node[plainNodes] (3) at (3,4) {$v_2$};
\node[plainNodes] (4) at (3,2.4) {$v_3$};
\node[plainNodes] (5) at (6,4) {$v_b$};
\node[plainNodes] (6) at (3.7,1) {$v_s$};

\draw[->] (0) -- node[labelNodes,sloped,above] {$\frac{1}{2}(1{-}\alpha_0)$} (2);
\draw[->] (0) -- node[labelNodes,sloped,above] {$\frac{1}{4}(1{-}\alpha_0)$} (3);
\draw[->] (0) -- node[labelNodes,sloped,above] {$\frac{1}{4}(1{-}\alpha_0)$} (4);
\draw[->] (2) to[bend left] node[labelNodes,sloped,above] {$\frac{1}{2}(1{-}\alpha_0)$} (0);
\draw[->] (2) to[bend right] node[labelNodes,sloped,above] {$\frac{1}{2}(1{-}\alpha_0)$} (5);
\draw[->,deletedTrans] (3) -- node[labelNodes,sloped,above] {$1{-}\alpha_0$} (5);
\draw[->, addedTrans] (3) to[bend right] node[labelNodes,above,pos=0.3] {$1{-}\alpha_0$} (0);
\draw[->] (4) -- node[labelNodes,left] {$1{-}\alpha_0$} (3);
\draw[->] (5) edge[loop right] node[labelNodes,] {$1$} (5);

\draw[->] (6) edge [loop right] node[labelNodes,] {$1$} (6);

\draw[->] (0) -- node[labelNodes, below, pos=0.8]{$\alpha_0$} (6);
\draw[->] (2) -- node[labelNodes, left, pos=0.5]{$\alpha_0$} (6);
\draw[->] (3.south east) -- node[labelNodes, right, pos=0.8]{$\alpha_0$} (6);
\draw[->] (4) -- node[labelNodes, left, pos=0.5]{$\alpha_0$} (6);
\end{tikzpicture}
\caption{Before entering an BSCC: merging sinks and BSCCs}
\label{fig:step4}
\end{figure}

Regarding this procedure, we present the following lemma, which claims that the probability of reaching $v_b$ from $v_0$ does not increase after each iteration. 
\begin{lemma}[Non-Increasing Reachability]\label{lm:decreasing_reachability}
For each loop iteration, we denote with $\Pr$ and $\Pr'$ the probabilities associated with transition function $\mcprob$ and $\mcprob'$ respectively, and it holds that $\Pr'^{v_0}(\Reach(v_b)) \le \Pr^{v_0}(\Reach(v_b))$.
\end{lemma}
\begin{proof}
Without loss of generality, we assume that there is no transition $(v_f,v_p)$ before the loop iteration. Otherwise, we can consider it as a separate transition from the newly added one.

\begin{itemize}
\item Before the loop iteration, we can classify the finite paths from $v_0$ to $v_b$ into two sets, namely: 
\begin{itemize}
    \item the paths that reach $v_b$ via the transition ($v_f, v_b)$;
    \item the paths that reach $v_b$ via a transition ($v'_f, v_b)$, where $ v'_f \neq v_f$.
\end{itemize} 
We write the finite paths from $v_0$ to $v_b$ formally as:
\[
\{ v_0 \cdots v_f v_b \} \uplus \{ v_0 \cdots v'_f v_b\ |\ v'_f \neq v_f \}
\]
Note that before the loop iteration, the transition $(v_f,v_p)$ is not available in either case, so we also write them as:
\[
\{ \overbrace{v_0 \cdots}^{\text{no }v_f v_p}  v_f v_b \} \uplus \{ \overbrace{v_0 \cdots}^{\text{no } v_f v_p} v'_f v_b\ |\ v'_f \neq v_f \}
\]
\item After the iteration, we can classify the finite paths from $v_0$ to $v_b$ into two sets:
\begin{itemize}
    \item the paths that take the transition $(v_f,v_p)$ at least once;
    \item the paths never take the transition $(v_f,v_p)$.
\end{itemize}
We write them formally as:
\[
    \{ \overbrace{v_0 \cdots}^{\textit{no } v_f v_p} v_f v_p \cdots v_b\} \uplus \{\overbrace{v_0 \cdots}^{\textit{no } v_f v_p} v'_f v_b\ |\  v'_f \neq v_f \}
\]
Note that in the latter set $v'_f \neq v_f$ is ensured implicitly since the transition $(v_f,v_b)$ has been removed.
\end{itemize}
Therefore the difference in reachability probabilities lies in the former set of finite paths, and we obtain the following:
\begin{align*}
  & \Pr'^{v_0}(\Reach(v_b)) - \Pr^{v_0}(\Reach(v_b)) \\
=\ & \Pr'(\{ \overbrace{v_0 \cdots}^{\textit{no } v_f v_p} v_f v_p \cdots v_b\}) - \Pr(\{ \overbrace{v_0 \cdots}^{\textit{no } v_f v_p} v_f v_b \}) \\
=\ & \Pr'(\{\pi\ |\ \pi = \overbrace{v_0\dots}^{\textit{no } v_f v_p} v_f\}) \cdot \mcprob'(v_f, v_p) \cdot \Pr'(\{\pi\ |\ \pi = v_p\dots v_f\}) \\
  & - \Pr(\{\pi\ |\ \pi=  v_0 \dots v_f\}) \cdot \mcprob(v_f,v_b) \\
=\ & \Pr(\{\pi\ |\ \pi = v_0 \dots v_f\}) \cdot \mcprob(v_f,v_b) \cdot \Pr'(\{\pi\ |\ \pi = v_p\dots v_f\}) \\
  & - \Pr(\{\pi\ |\ \pi= v_0 \dots v_f\}) \cdot \mcprob(v_f,v_b) \\
=\ & \Pr(\{\pi\ |\ \pi = v_0 \dots v_f\}) \cdot \mcprob(v_f,v_b) \cdot (\Pr'(\{\pi\ |\ \pi = v_p\dots v_{f}\})-1) \\
\le\ & 0
\end{align*}
We thus conclude that in each loop iteration, the probability of reaching $v_b$ from $v_0$ does not increase.
\end{proof}

Furthermore, we apply a similar transformation as before to the last frontier vertex $v_f$ so that $\mcprob_4(v_f,v_b)=\pM\alpha_{0}$. The redundant probability is transferred to transition $(v_f,v_0)$. The transformation on the running example is demonstrated in Figure~\ref{fig:step5}, where $v_2$ is the last frontier vertex. We consider the transition $(v_2,v_b)$ as two separate ones, and substitute one of them with a transition $(v_2,v_0)$ for clarity. Note that the deleted and added transitions are drawn in orange and purple respectively.
\begin{figure}
\centering
\begin{tikzpicture}[scale=1.5]
\node[plainNodes] (0) at (0,4) {$v_0$};
\node[plainNodes] (2) at (3,5) {$v_2$};
\node[plainNodes] (5) at (6,4) {$v_b$};
\draw[->] (2) -- node[labelNodes,above,sloped] {$\frac{1}{2}(1{-}\alpha_0)$} (0);
\draw[->,addedTrans] (2) to[bend left] node[labelNodes,sloped, below] {$(\frac{1}{2}{-}\pM)(1{-}\alpha_0)$} (0);
\draw[->] (2) -- node[labelNodes,sloped,above] {$\pM(1{-}\alpha_0)$} (5);
\draw[->,deletedTrans] (2) to[bend right] node[labelNodes,sloped,below] {$(\frac{1}{2}{-}\pM)(1{-}\alpha_0)$} (5);

\end{tikzpicture}
\caption{Before entering an rBSCC: the final action}
\label{fig:step5}
\end{figure}
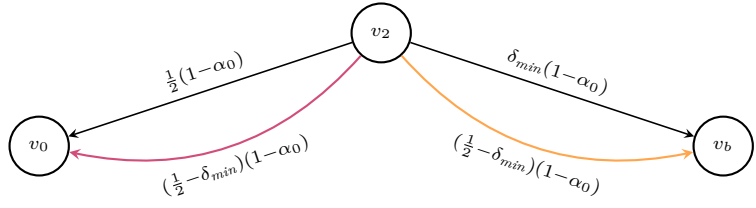
As this transformation is essentially the same as `defrontierizing', we obtain that the probability of reaching $v_b$ from $v_0$ does not increase through a similar analysis to the proof of Lemma~\ref{lm:decreasing_reachability}.

After applying the procedure, only one frontier vertex remains, and we denote it with $v_f$. The resulting Markov Chain is $\mc_4 = (V_{T}\uplus\{v_s,v_b\}, \mcprob_4)$. It follows from Lemma~\ref{lm:decreasing_reachability} that:
\[\Pr^{v_0}_3(\Reach(v_b)) \ge \Pr^{v_0}_4(\Reach(v_b))\]

\begin{observation}\label{obs:stpe4}
\begin{itemize}
\item It follows from previous observations and Remark~\ref{rm:all_can_reach} that for all vertices $v\in V_T\backslash\{v_f\}$, we have $\Pr_4^{v}(\Reach(v_{b}))>0$.
\item It follows from previous transformation steps that:
\begin{itemize}
\item $\{v_s\}$ and $\{v_b\}$ are the only BSCCs, and $\mcprob_4(v_s,v_s)= \mcprob_4(v_b,v_b) =1$.
\item For all vertices $v\in V\backslash\{v_f\}$, we have $\mcprob_4(v,v_s) = \alpha_{0}$, $\mcprob_4(v,V_T)=1-\alpha_{0}$ and $\mcprob_4(v,v_b)=0$.
\item For all $v,v'\in V_T\times V_T$, $\mcprob_4(v,v')>0$ implies that $\mcprob_4(v,v')>\pM(1-\alpha_{0})$.
\end{itemize}
\item For the frontier vertex, we have $\mcprob_4(v_f,v_b)=\pM\alpha_{0}$ and $\mcprob_4(v_f,V_T)=\pM(1-\alpha_{0})$.
\end{itemize}
\end{observation}
\end{enumerate}

Combining the reasoning of the $4$-step transformation, we obtain that for all $\ov \in \oV_T$, $\wPr^{\ov}_{\sigma,\gamma}(\crossPath) \ge \Pr^{v}_4(\Reach(v_b))$, which concludes the first part of the proof of Lemma~\ref{lm:formal-crosspath}.

\subsubsection{Second Part: Applying Lemma~\ref{lm:lower_bound_general}} 

We apply Lemma~\ref{lm:lower_bound_general} to $\mc_4 = (V_{T}\uplus\{v_s,v_b\}, \mcprob_4)$ from to obtain a lower bound of $\wPr^{\ov}_{\sigma,\gamma}(\crossPath)$.
It follows from Observation~\ref{obs:stpe4} that Lemma~\ref{lm:lower_bound_general} can be applied to $\mc_4$, where:
\begin{itemize}
    \item $s=l=\pM(1-\alpha_{0})$
    \item $k=t=(1-\pM)(1-\alpha_{0})$
\end{itemize}
Therefore we obtain a lower bound of $\Pr^{v}_4(\Reach(v_b))$ as follows:
\[
 \Pr^{v}_4(\Reach(v_b)) \ge \frac{(1-x_0)x_0^n}{(1-x_0)-(1-x_0^n)x_1} 
\]
where $x_0 = \pM(1-\alpha_{0})$ and $x_1=(1-\pM)(1-\alpha_{0})$. We thus obtain a lower bound of $\wPr^{\ov}_{\sigma,\gamma}(\crossPath)$ as well.

\subsection{Proof of Lemma~\ref{lm:formal-wineven}}
\label{app-subsec:wineven}
\lowerwineven*

To prove Lemma~\ref{lm:formal-wineven},
we fix an arbitrary pair of strategies $\sigma,\gamma\in\Sigma\times\Gamma$, and obtain as before the Markov chains $\mc_{\sigma,\gamma}$ and $\omc_{\sigma,\gamma}$. We again leave out the initial state. Hence we have $\mc_{\sigma,\gamma}=(V,\mcprob)$ and $\omc_{\sigma,\gamma} = (\oV\uplus \widehat{V}\uplus\{\vwin,\vlose\}, \omcprob)$. 

We consider an arbitrary even BSCC $C$ in $\mc_{\sigma,\gamma}$, but the numeric results would be the exact same in an odd BSCC, replacing the winning sink with a losing sink. We denote the smallest priority of $C$ with $k$. The counterpart of $C$ in $\omc_{\sigma,\gamma}$ is denoted with $\oC\uplus\widehat{C}$. It follows from the construction of $\wG$ that no transitions are leaving $\oC\uplus\widehat{C}$ except for entering the winning and losing sinks. Therefore, $\oC\uplus\widehat{C}\uplus\{\vwin,\vlose\}$ and their internal transitions can be viewed as an independent Markov Chain, and we denote it with $\omc_c=(\oC\uplus\widehat{C}\uplus\{\vwin,\vlose\},\omcprob_c)$. We denote all vertices of priority $k$ in $C$ with $C_k$, and let $C_{>k} = C\backslash C_k$.

Similar to the proof of Lemma~\ref{lm:formal-crosspath}, the proof of Lemma~\ref{lm:formal-wineven} also consists of two parts:
\begin{enumerate}
\item We make a $4$-step transformation to $\omc_c$. 
We denote the resulting Markov Chain of the $i$-th step with $\mc_{ci} = (V_{ci}, \mcprob_{ci})$. In the $i$-th Markov Chain, we denote with $\Pr^v_{ci}(\Reach(V'_i))$ the probability of reaching $V'_i\subseteq V_i$ from $v$. We show that for all $\ov\in\oC$, we have:
\[\wPr^{\ov}_{\sigma,\gamma}(\Reach(\vwin)) \ge \Pr^{v}_{c4}(\Reach(\vwin))\]
\item We apply Lemma~\ref{lm:lower_bound_general} to $\mc_{c4}$ to obtain a lower bound of $\Pr^{v}_{c4}(\Reach(\vwin))$, and hence a lower bound of $\wPr^{k}_{\sigma,\gamma}(\winEven)$. 
\end{enumerate}
We now present the proof as follows.
\subsubsection{First Part: Transforming the Markov Chain}
\begin{enumerate}
\item We first apply the same transformation as~\ref{step1} of the previous proof to $\omc_c$, eliminating $\widehat{C}$ and renaming all $\ov\in\oC$ as $v$. We have $\mc_{c1} = (C\uplus\{\vwin,\vlose\}, \mcprob_{c1})$, and for all $\ov\in \oC$, it holds that $\wPr^{\ov}_{\sigma,\gamma}(\Reach(\vwin)) = \Pr^{v}_{c1}(\Reach(\vwin))$.
\item For all vertex pairs $u,v\in C\times C_{>k}$ with $p(v)=2m$ for some $m\in\mathbb{N}$, we transfer probability $\mcprob(u,v)\cdot\alpha_{2m}$ from transition $(u,\vwin)$ to transition $(u,\vlose)$. 
We give Figure~\ref{fig:insideRBSCCs-2} as an illustration. The deleted and added transitions are drawn in orange and purple respectively, and the transition $(u,\vlose)$ has probability $y+\mcprob(u,v)\cdot\alpha_{2m}$ with two separate ones combined. 
It follows that for all $v\in C$, we have $\Pr^{v}_{c1}(\Reach(\vwin)) \ge \Pr^{v}_{c2}(\Reach(\vwin))$.
\begin{figure}[H]
\centering
\begin{tikzpicture}[scale=1.2,
]
\node[plainNodes] (0) at (0,0) {$u$};
\node[plainNodes,label=above:{\scriptsize $p(v)=2m>k$}] (1) at (-4,1) {$v$};
\node[winNode] (2) at (4,0) {$\vwin$};
\node[loseNode] (3) at (0,-3) {$\vlose$};
\draw[->] (0) -- node[labelNodes,sloped,below] {$\mcprob(u,v)(1{-}\alpha_{2m})$} (1);

\draw[->] (0) -- node[labelNodes,above] {$x{-}\mcprob(u,v)\alpha_{2m}$} (2);
\draw[->,deletedTrans] (0) to[bend right] node[labelNodes,below] {$\mcprob(u,v)\alpha_{2m}$} (2);
\draw[->,addedTrans] (0) to[bend left] node[labelNodes,right] {$\mcprob(u,v)\alpha_{2m}$} (3);
\draw[->] (0) -- node[labelNodes,left] {$y$} (3);

\end{tikzpicture}
\caption{Inside a BSCC: the second step}
\label{fig:insideRBSCCs-2}
\end{figure}
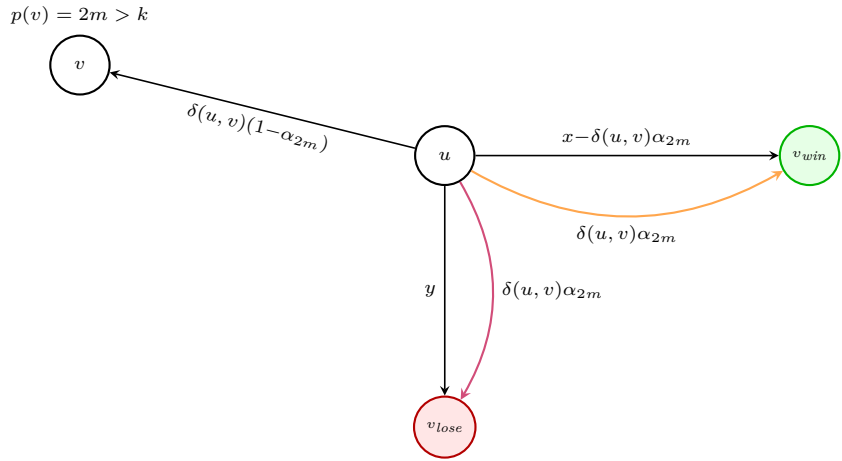
After the second step, we observe that all transitions entering $\vwin$ have the probability $(\sum_{v\in C_k}\mcprob(u,v))\cdot\alpha_{k}$ for some $u\in C$.

\item We scale up $\alpha$ in a similar manner as in step~\ref{step2} of Section~\ref{app-subsec:transform}. For all integers $n\ge k+1$, we let $\alpha_{n}=\alpha_{k+1}$. It follows that for all~$v\in C$, the probability of reaching $\vlose$ from $v$ does not decrease, and thus we have $\Pr^{v}_{c2}(\Reach(\vwin)) \ge \Pr^{v}_{c3}(\Reach(\vwin))$.
\item We look at an vertex pair $v,v_k\in C\times C_k$. 
\begin{figure}[H]
\centering
\begin{tikzpicture}[scale=1.4,
]


\node[plainNodes] (00) at (0,-6) {$v$};
\node[plainNodes] (01) at (-4,-5) {$v_k$};
\node[winNode] (02) at (4,-6) {$\vwin$};
\node[loseNode] (03) at (0,-8) {$\vlose$};
\draw[->] (00) -- node[labelNodes,sloped,below] {$\mcprob(v,v_k)(1{-}\alpha_{k})+$\textcolor{purple}{$\mcprob(v,v_k)(\alpha_k {-} \alpha_{k+1})$}} (01);
\draw[->] (00) -- node[labelNodes,above] {$x{-}\mcprob(v,v_k)\alpha_k$} (02);
\draw[->,deletedTrans] (00) to[bend right] node[labelNodes,below] {$\mcprob(v,v_k)\alpha_k$} (02);
\draw[->] (00) -- node[labelNodes,left] {$y+$ \textcolor{purple}{$\mcprob(v,v_k)\alpha_{k+1}$}} (03);

\end{tikzpicture}
\caption{Inside a BSCC: the fourth step}
\label{fig:insideRBSCCs-4}
\end{figure}
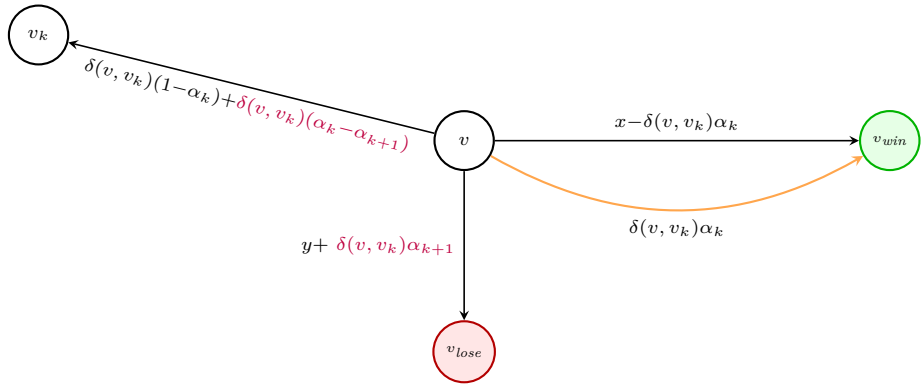
If we make the transformation given in Figure~\ref{fig:insideRBSCCs-4}, with a similar analysis as in the proof of Lemma~\ref{lm:decreasing_reachability}, we obtain the following for all $v_0\in C$, where $\Pr$ and $\Pr'$ are the probabilities associated with the transition functions before and after the transformation respectively.
\begin{align*}
 & \Pr'^{v_0}(\Reach(\vwin)) - \Pr^{v_0}(\Reach(\vwin)) \\
=\ &  \Pr(\{\pi\ |\ \pi=v_0\cdots v\})\mcprob(v,v_k)(\alpha_{k}-\alpha_{k+1})\Pr'(\Reach(\vwin)) \\
   & - \Pr(\{\pi\ |\ \pi=v_0\cdots v\})\mcprob(v,v_k) \\
=\ & \Pr(\{\pi\ |\ \pi=v_0\cdots v\})\mcprob(v,v_k)((\alpha_{k}-\alpha_{k+1})\Pr'(\Reach(\vwin))-1) \\
\le\ & 0
\end{align*}
We apply this transformation to all but one vertex pairs $v,v_k\in C\times C_k$. After this there exists only one vertex pair $v,v_k\in C\times C_k$, where transition $(v,v_k)$ and transition $(v,\vwin)$ have probabilities $x(1-\alpha_{k})$ and $x\alpha_{k}$ respectively. For clarity, we consider these two transitions as four, and apply the transformation given in Figure~\ref{fig:insideRBSCCs-5} to finish the last step.  
\begin{figure}[H]
\centering
\begin{tikzpicture}[scale=1.2,
]


\node[plainNodes] (00) at (0,-6) {$v$};
\node[plainNodes] (01) at (-4,-5) {$v_k$};
\node[winNode] (02) at (4,-6) {$\vwin$};
\node[loseNode] (03) at (0,-8.5) {$\vlose$};
\draw[->] (00) -- node[labelNodes,sloped,above] {$\pM(1{-}\alpha_k)$} (01);
\draw[->] (00) to[bend left] node[labelNodes,pos=0.1, anchor=south east, yshift=-0.5cm,xshift=0.3cm] {$(x{-}\pM)(1{-}\alpha_{k})+$\textcolor{purple}{$(x{-}\pM)(\alpha_k {-} \alpha_{k+1})$}} (01);
\draw[->] (00) -- node[labelNodes,above] {$\pM\alpha_k$} (02);
\draw[->, deletedTrans] (00) to[bend right] node[labelNodes,below] {$(x{-}\pM)\alpha_k$} (02);
\draw[->] (00) -- node[labelNodes, left] {$y+$ \textcolor{purple}{$(x{-}\pM)\alpha_{k+1}$}} (03);

\end{tikzpicture}
\caption{Inside a BSCC: the final action}
\label{fig:insideRBSCCs-5}
\end{figure}
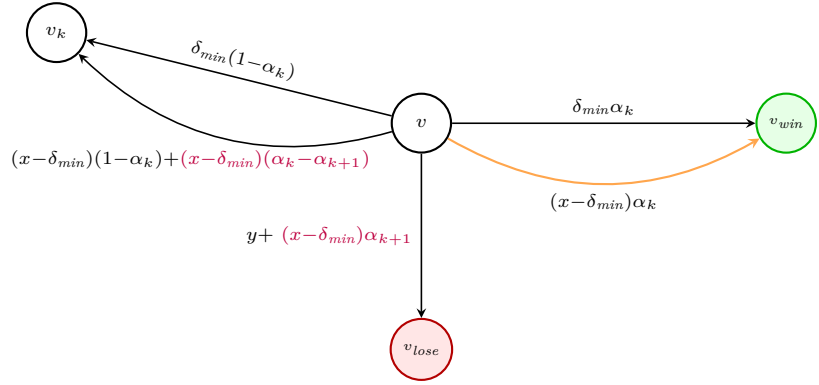
As a result, we observe that $\mcprob_{c4}(v,\vwin)=\pM\alpha_{k}$ and $\mcprob_{c4}(v,\vlose)=(1-\pM)\alpha_{k+1}$. As this last action is essentially the same as the transformation in Figure~\ref{fig:insideRBSCCs-4}, we claim that for all $v\in C$, we have $\Pr^{v}_{c3}(\Reach(\vwin)) \ge \Pr^{v}_{c4}(\Reach(\vwin))$. We also make the following observation regarding $\mc_{c4}$.

\begin{observation}\label{obs:insiderbscc}
\begin{itemize}
\item There are only two BSCCs in $\mc_{c4}$, namely $\{\vwin\}$ and $\{\vlose\}$, and $\mcprob_{c4}(\vwin)=\mcprob_{c4}(\vlose)=1$.
\item There is only one vertex, denoted with $v_f\in C$, such that $\mcprob_{c4}(v_f,\vwin)>0$. We also have $\mcprob_{c4}(v_f,C) = \pM(1-\alpha_{k})+(1-\pM)(1-\alpha_{k+1})$ and $\mcprob_{c4}(v_f,\vwin) = \pM\alpha_{k}$.
\item Since we never change the connectivity inside $C$ during the transformations, for all $v \in C$, we have $\Pr^{v}_{c4}(\Reach(\vwin))>0$. For all $v\in C\backslash\{v_f\}$, we have $\mcprob_{c4}(v,\vlose)=\alpha_{k+1}$, $\mcprob_{c4}(v,C)=1-\alpha_{k+1}$ and $\mcprob_{c4}(v,\vwin)=0$; for all $v'\in C$, $\mcprob_{c4}(v,v')>0$ implies that $\mcprob_{c4}(v,v')>\pM(1-\alpha_{k+1})$.
\end{itemize}
\end{observation}
\end{enumerate}
\subsubsection{Second Part: Applying Lemma~\ref{lm:lower_bound_general}} 
When applying Lemma~\ref{lm:lower_bound_general} to $\mc_{c4}$ from the previous part we obtain a lower bound of $\wPr^{\ov}_{\sigma,\gamma}(\Reach(\vwin))$. To obtain a general lower bound for all $\wv\in \oV\uplus\widehat{V}$, we multiply it by a factor $(1-\alpha_{k+1})$. For the case of an odd pBSCC, we would obtain the exact same lower bound, but for reaching the losing sink, and so an upper bound on the probability of winning would be $1-\wPr^{\ov}_{\sigma,\gamma}(\Reach(\vlose))$.

\subsection{Proof of Lemma~\ref{lm:interval}}
\label{app-subsec:interval}

\interval*

For all vertices $v \in V$, we make the following observations:
\begin{enumerate}
    \item\label{aaa} $\Pr_{\sigma,\gamma}^v(\enterEven)+\Pr_{\sigma,\gamma}^v(\enterOdd) = 1$
    \item\label{aab} $\wPr^{\ov}_{\sigma,\gamma}(\enterEven)+\wPr^{\ov}_{\sigma,\gamma}(\enterOdd) = \wPr^{\ov}_{\sigma,\gamma}(\crossPath)$
    \item\label{bbb} $\wPr_{\sigma,\gamma}^\ov(\enterEven) \le \Pr_{\sigma,\gamma}^v(\enterEven)$ and $\wPr_{\sigma,\gamma}^\ov(\enterOdd) \le \Pr_{\sigma,\gamma}^v(\enterOdd)$
    \item\label{ccc} Regarding $\wPr_{\sigma,\gamma}^\ov(\enterEven)$, we obtain the following:
    \begin{align*}
       & \wPr_{\sigma,\gamma}^\ov(\enterEven) \\ 
    =\ & \wPr_{\sigma,\gamma}^\ov(\crossPath) - \wPr_{\sigma,\gamma}^\ov(\enterOdd) & \textit{by~\ref{aab}} &\\
    \ge\ & \wPr_{\sigma,\gamma}^\ov(\crossPath) -\Pr_{\sigma,\gamma}^v(\enterOdd) & \textit{by~\ref{bbb}} & \\
    =\ & \wPr_{\sigma,\gamma}^\ov(\crossPath) + \Pr_{\sigma,\gamma}^v(\enterEven) - 1 & \textit{by~\ref{aaa}} &
    \end{align*}
\end{enumerate}

We thus obtain a lower bound of $\wPr_{\sigma,\gamma}^\ov(\Reach(\vwin))$ as follows:
\begin{align*}
        & \wPr_{\sigma,\gamma}^\ov(\Reach(\vwin)) \\ 
    \ge\ & \wPr_{\sigma,\gamma}^\ov(\enterEven) \cdot \wPr_{\sigma,\gamma}^{k}(\winEven) \\
    \ge\ & (\wPr_{\sigma,\gamma}^\ov(\crossPath) + \Pr_{\sigma,\gamma}^v(\enterEven) - 1) \cdot \wPr_{\sigma,\gamma}^{k}(\winEven) \quad & \tag*{\textit{by~\ref{ccc}}} \\
    \ge\ & (x+\Pr_{\sigma,\gamma}^v(\enterEven)-1)\cdot y \quad  \\
    =\   & y \cdot \Pr_{\sigma,\gamma}^v(\enterEven) - y + x\cdot y
\end{align*}

Similarly, we obtain an upper bound of $\wPr_{\sigma,\gamma}^\ov(\Reach(\vwin))$ as follows:
\begin{align*}
       & \wPr_{\sigma,\gamma}^\ov(\Reach(\vwin))  \\
    <\ & (1-\wPr_{\sigma,\gamma}^\ov(\crossPath)) + \wPr_{\sigma,\gamma}^\ov(\enterEven)\cdot 1 + \wPr_{\sigma,\gamma}^\ov(\enterOdd) \cdot \wPr_{\sigma,\gamma}^{k}(\winOdd) \\
    \le\ & (1-\wPr_{\sigma,\gamma}^\ov(\crossPath))+\Pr_{\sigma,\gamma}^v(\enterEven) \cdot 1 + \\
         & (\wPr_{\sigma,\gamma}^\ov(\crossPath)-\Pr_{\sigma,\gamma}^v(\enterEven)) \cdot \wPr_{\sigma,\gamma}^{k}(\winOdd) \tag*{\textit{by~\ref{aab}}}\\
    =\ & (1-\wPr_{\sigma,\gamma}^{k}(\winOdd))\cdot \Pr_{\sigma,\gamma}^v(\enterEven) + 1 - \wPr_{\sigma,\gamma}^\ov(\crossPath)\cdot(1-\wPr_{\sigma,\gamma}^{k}(\winOdd)) \\
    \le\ & \Pr_{\sigma,\gamma}^v(\enterEven) + 1-x\cdot y 
\end{align*}
Therefore we get a range of $\wPr_{\sigma,\gamma}^\ov(\Reach(\vwin))$ as follows:
\[
y \cdot \Pr_{\sigma,\gamma}^v(\enterEven) - y+x\cdot y \ \le \ \wPr_{\sigma,\gamma}^\ov(\Reach(\vwin)) \ \le\ \Pr_{\sigma,\gamma}^v(\enterEven) + 1-x\cdot y
\]

By Corollary~\ref{corollary:pms}, we have:
\begin{itemize}
    \item In the SPG $(G,\gparity(p))$, $\arnpr^{v}_{\sigma,\gamma}(\gparity(p)) = \Pr^{v}_{\sigma,\gamma}(\enterEven)$.
    \item In the SSG $(\wG, \greach(\vwin))$, $\oarnpr^{\ov}_{\sigma,\gamma}(\greach(\vwin)) = \wPr^{\ov}_{\sigma,\gamma}(\Reach(\vwin))$.
\end{itemize}
Hence we conclude:
\[
y\cdot \arnpr^{v}_{\sigma,\gamma} - y+x\cdot y \  \leq \
\oarnpr^{\ov}_{\sigma,\gamma} \ \leq \
\arnpr^{v}_{\sigma,\gamma} + 1-x \cdot y
\]

\subsection{Proof of Lemma~\ref{lm:ar-alpha}}
\label{app-subsec:ar-alpha}

\aralpha*

\subsubsection{Arranging $\alpha_{0}$}

We start by giving a bound on $\wPr^{\ov}_{\sigma,\gamma}(\crossPath)$ probability, that involves $\alpha_0$. To do so, we make use of Lemma~\ref{lm:formal-crosspath}. 
\begin{corollary}[Another Lower Bound of $\crossPath$ Probability]  \label{coro:actual-crosspath}
For all strategy pairs $\sigma,\gamma\in\Sigma_{\exists}\times\Sigma_{\forall}$, for all $\ov\in\oV$, the following holds:
\begin{equation}
\wPr^{\ov}_{\sigma,\gamma}(\crossPath) > \frac{\pM^n(1-\alpha_0)^{n+1}}{2\alpha_0+\pM^n(1-\alpha_0)^{n+1}}\notag
\end{equation}
\end{corollary}

\begin{proof}
We scale down $\wPr^{\ov}_{\sigma,\gamma}(\crossPath)$ as follows: 
\begin{align*}
& \wPr^{\ov}_{\sigma,\gamma}(\crossPath) \\
\ge\ & \frac{(1-x_0)x_0^n}{(1-x_0)-(1-x_0^n)x_1} \tag*{\textit{Lemma~\ref{lm:formal-crosspath}}}\\
=\ & \frac{(1-\pM(1-\alpha_0))\pM^n(1-\alpha_0)^n}{\alpha_0+(1-\pM)\pM^n(1-\alpha_0)^{n+1}} \\
>\ & \frac{(1-\pM)(1-\alpha_0)\pM^n(1-\alpha_0)^n}{\alpha_0+(1-\pM)\pM^n(1-\alpha_0)^{n+1}} \tag*{\textit{since $1-\pM(1-\alpha_0)>(1-\pM)(1-\alpha_0)$}}\\
=\ & \frac{\pM^n(1-\alpha_0)^{n+1}}{\frac{\alpha_0}{1-\pM}+\pM^n(1-\alpha_0)^{n+1}} \\
\ge\ & \frac{\pM^n(1-\alpha_0)^{n+1}}{2\alpha_0+\pM^n(1-\alpha_0)^{n+1}} \tag*{\textit{since $1-\pM\ge\frac{1}{2}$}}
\end{align*}
\end{proof}

We can now arrange $\alpha_0$. It follows from Corollary~\ref{coro:actual-crosspath} that:
\begin{equation*}
\wPr^{\ov}_{\sigma,\gamma}(\crossPath) > \frac{\pM^n(1-\alpha_0)^{n+1}}{2\alpha_0+\pM^n(1-\alpha_0)^{n+1}}
\end{equation*}
Therefore to show:
\begin{equation}\label{eq:qcrosspath-1}
\wPr^{\ov}_{\sigma,\gamma}(\crossPath) > \frac{4-\epsilon}{4}
\end{equation}
where $\epsilon = \frac{1}{\epsi}$, it suffices to show that:
\begin{equation}
\frac{\pM^n(1-\alpha_0)^{n+1}}{2\alpha_0+\pM^n(1-\alpha_0)^{n+1}} \ge \frac{4-\epsilon}{4}
\end{equation}
which can be further simplified as:
\begin{equation}\label{eq:qcrosspath-2}
\epsilon \ge \frac{8\alpha_0}{2\alpha_0+\pM^n(1-\alpha_0)^{n+1}}
\end{equation}
We show that when $\alpha_0\le\frac{\pM^n}{8\epsi}$, inequality~\ref{eq:qcrosspath-2} holds. We start with the right side:
\begin{align*}
     & \frac{8\alpha_0}{2\alpha_0+\pM^n(1-\alpha_0)^{n+1}} \\
\le\ & \frac{8\alpha_0}{2\alpha_0+\pM^n(1-(n+1)\alpha_0)} \tag*{\textit{follows from Bernoulli's inequality}} \\
=\ & \frac{8\alpha_0}{\alpha_0+\pM^n+\alpha_0({1-(n+1)\pM^n})} \\
\le\ & \frac{8\alpha_0}{\alpha_0+\pM^n} \tag*{\textit{since $1-(n+1)\pM^n>0$ }} \\
\le\ & \frac{\frac{\pM^n}{\epsi}}{\frac{\pM^n}{8\epsi}+\pM^n} \\
=\ & \frac{1}{\frac{1}{8}+\epsi} \\
\le\ & \frac{1}{\epsi} = \epsilon
\end{align*}
Therefore we obtain that when $\alpha_{0}\le\frac{\pM^n}{8\epsi}$, inequality~\ref{eq:qcrosspath-1} holds, and thus condition~\ref{aspt:c0} is satisfied.

\subsubsection{Arranging $\alpha_{k+1}/\alpha_{k}$}

We start by getting a lower bound on $\wPr_{\sigma,\gamma}^{k}(\winEven)$ for even $k$'s, which makes use of $\alpha_{k+1}$ and $\alpha_k$. The reasoning for odd $k$'s is symmetric. To do so, we use Lemma~\ref{lm:formal-wineven}.

\begin{corollary}[Another Lower Bound of $\winEven$ Probability]  \label{coro:actual-wineven}
For all strategy pairs $\sigma,\gamma\in\Sigma\times\Gamma$, for all even $k$, the following holds:
\begin{equation}
\wPr_{\sigma,\gamma}^{k}(\winEven) > \frac{\pM^n(1-\pM)-\frac{\aKP}{\aK}}{\pM^n(1-\pM)+\frac{\aKP}{\aK}}\notag
\end{equation}
\end{corollary}

\begin{proof}
We scale down the right side as follows:

\allowdisplaybreaks
\begin{align*}
& \wPr_{\sigma,\gamma}^{k}(\winEven) \\
\ge & (1-\alpha_{k+1}) \cdot \frac{(1-x_2)\cdot x_2^{n-1}\cdot x_4}{1-(x_2+x_3) + x_5\cdot x_2^{n} + t\cdot x_2^{n-1} - x_5 \cdot x_2^{n-1}} \tag*{\textit{Lemma~\ref{lm:formal-wineven}}}\\
=\ & \frac{(1-\pM+\pM\aKP) \cdot \pM^{n-1}(1-\aKP)^{n} \cdot \pM\aK }{\aKP+\pM^{n-1}(1-\aKP)^{n-1}(\pM^2(-\aK+\aKP+\aK\aKP-\aKP^2)+\pM(\aK-2\aKP+\aKP^2))} \\
>\ & \frac{(1-\pM) \cdot \pM^{n}\cdot(1-\aKP)^{n+1} \cdot \aK}{\aKP+\pM^{n}\cdot(\pM(-\aK+\aKP+\aK\aKP-\aKP^2)+(\aK-2\aKP+\aKP^2))} \tag*{\textit{since $1-\pM+\pM\aKP>(1-\pM)(1-\aKP)$ and $(1-\aKP)^{n-1} < 1$}} \\
\ge\ & \frac{(1-\pM)\pM^n\aK(1-(n+1)\aKP)}{\aKP+\pM^n(\aK(1-\pM)-\aKP(2-\aKP-\pM(1+\aK-\aKP)))} \tag*{\textit{follows from Bernoulli inequality}}\\
>\ & \frac{(1-\pM)\pM^n\aK-(n+1)\pM^n(1-\pM)\aK\aKP}{\aKP+(1-\pM)\pM^n\aK} \tag*{\textit{since $\aKP(2-\aKP-\pM(1+\aK-\aKP))>0$}}\\
>\ & \frac{\pM^n\aK(1-\pM)-\aKP}{\pM^n\aK(1-\pM)+\aKP} \tag*{\textit{since $(n+1)\pM^n(1-\pM)\aK<1$}} \\
=\ & \frac{\pM^n(1-\pM)-\frac{\aKP}{\aK}}{\pM^n(1-\pM)+\frac{\aKP}{\aK}}
\end{align*}
\end{proof}

We now arrange $\alpha_{k+1}/\alpha_{k}$. It follows from Corollary~\ref{coro:actual-wineven} that for all strategy pairs $\sigma,\gamma\in\Sigma\times\Gamma$, for all even $k$, the following holds:
\begin{equation}
\wPr_{\sigma,\gamma}^{k}(\winEven) > \frac{\pM^n(1-\pM)-\frac{\aKP}{\aK}}{\pM^n(1-\pM)+\frac{\aKP}{\aK}}
\end{equation}
Therefore to show:
\begin{equation}\label{eq:qwineven-0}
\wPr_{\sigma,\gamma}^{k}(\winEven)\ge \frac{4}{4+\epsilon}
\end{equation}
it suffices to show that for all $k\in\mathbb{N}$:
\begin{equation}\label{eq:qwineven-1}
\frac{\pM^n(1-\pM)-\frac{\aKP}{\aK}}{\pM^n(1-\pM)+\frac{\aKP}{\aK}} \ge \frac{4}{4+\epsilon}
\end{equation}
We denote $\frac{\aKP}{\aK}$ with $r$ in the following calculation. Inequality~\ref{eq:qwineven-1} can be further simplified to:
\begin{equation}
\frac{\epsilon}{4+\epsilon} \ge \frac{2r}{(1-\pM)\pM^n+r}
\end{equation}
and can be finally simplified to:
\begin{equation}
r \le \frac{(1-\pM)\pM^n}{8\epsi+1}
\end{equation}
Therefore we obtain that if for all $k\in\mathbb{N}$, the following holds:
\begin{equation}
\frac{\alpha_{k+1}}{\alpha_{k}} \le \frac{(1-\pM)\pM^n}{8\epsi+1} \notag
\end{equation}
then inequality~\ref{eq:qwineven-0} holds, and thus condition~\ref{aspt:c1} is satisfied.

\end{document}